\documentclass[twoside,10pt,a4paper]{article}

\usepackage[utf8]{inputenc}
\usepackage[T1]{fontenc}
\usepackage{csquotes}
\usepackage[english]{babel}
\usepackage{amsfonts}
\usepackage{amsmath}
\usepackage{amsthm}
\usepackage{amssymb}
\usepackage{faktor } 
\usepackage{titlesec}
\usepackage{hyperref} 
\usepackage{MnSymbol} 
\usepackage{graphicx}
\usepackage{bm}
\usepackage{subcaption}
\usepackage{dsfont}
\usepackage{tikz}
\usetikzlibrary{shapes}
\usepackage{pgfplots}
\usetikzlibrary{calc}
\usetikzlibrary{matrix,arrows}
\usepackage{graphicx}

\newtheorem{theo}{Theorem}

\newcounter{nb}[subsection]

\newtheorem{theorem}[nb]{Theorem}

\newtheorem{lemma}[nb]{Lemma}
\newtheorem{prop}[nb]{Proposition}
\newtheorem{corollary}[nb]{Corollary}

\theoremstyle{definition}
\newtheorem{definition}[nb]{Definition}
\newtheorem{notation}[nb]{Notation}

\theoremstyle{example}
\newtheorem{example}[nb]{Example}

\theoremstyle{remark}
\newtheorem{remark}[nb]{Remark}
\newtheorem{obs}[nb]{Observation}

\usepackage{scrextend}
 
\DeclareMathOperator{\ev}{ev}
\DeclareMathOperator{\bideg}{bideg}

\DeclareMathOperator{\Pic}{Pic}

\DeclareMathOperator{\LT}{LT}
\DeclareMathOperator{\LM}{LM}

\DeclareMathOperator{\Span}{Span}
\DeclareMathOperator{\range}{range}
\DeclareMathOperator{\Id}{Id}

\newcommand{\A}{\mathbb{A}}
\newcommand{\Gb}{\mathbb{G}}
\newcommand{\R}{\mathbb{R}}

\newcommand{\Z}{\mathbb{Z}}
\newcommand{\F}{\mathbb{F}}
\newcommand{\PP}{\mathbb{P}}
\newcommand{\N}{\mathbb{N}}
\newcommand{\Mon}{\mathcal{M}}

\renewcommand{\tilde}{\widetilde}

\title{Algebraic Geometric codes on minimal Hirzebruch surfaces}
\author{Jade Nardi \thanks{Institut de Math\'ematiques de Toulouse~; UMR 5219, Universit\'e de Toulouse~; CNRS UPS IMT, F-31062 Toulouse Cedex 9, France jade.nardi@math.univ-toulouse.fr.
Funded by ANR grant ANR-15-CE39-0013-01 "manta"}}

\begin{document}

\maketitle

\begin{abstract}
We define a linear code $C_\eta(\delta_T,\delta_X)$ by evaluating polynomials of bidegree $(\delta_T,\delta_X)$ in the Cox ring on $\F_q$-rational points of a minimal Hirzebruch surface over the finite field $\F_q$. We give explicit parameters of the code, notably using Gr\"obner bases. The minimum distance provides an upper bound of the number of $\F_q$-rational points of a non-filling curve on a Hirzebruch surface.
\end{abstract}



\section*{Introduction}

Until the $00$'s, most Goppa codes were associated to curves. In 2001 S.H. Hansen \cite{Han} estimated parameters of Goppa codes associated to normal projective varieties of dimension at least $2$. As Hansen required very few assumptions on the varieties, the parameters he gave depended only on the Seshadri constant of the line bundle, which is hard to compute in practice. New classes of error correcting codes have thus been constructed, focusing on specific well-known families of varieties to better grasp the parameters. Among Goppa codes associated to a surface which have been studied so far, some toric and projective codes are based on Hirzebruch surfaces.

\emph{Toric codes}, first introduced by J. P. Hansen \cite{HanToric} and further investigated by D. Joyner \cite{Joyner}, J. Little and H. Schenck \cite{Little}, D. Ruano \cite{Ruano} and I. Soprunov and J. Soprunova \cite{SS}, are Goppa codes on toric varieties evaluating global sections of a line bundle at the $\F_q$-rational points of the torus. J. Little and H. Schenck \cite{Little} already computed the parameters of toric codes on Hirzeburch surfaces for some bidegrees and for $q$ large enough to make the evaluation map injective.

\emph{Projective codes} evaluate homogeneous polynomials on the rational points of a variety embedded in a projective space. A first example of projective codes is the family of Reed-Muller projective codes on $\PP^n$ \cite{Lachaud}. A. Couvreur and I. Duursma \cite{CD} studied codes on the biprojective space $\PP^1\times \PP^1$ embedded in $\PP^3$. The authors took advantage of the product structure of the variety, yielding a description of the code as a tensor product of two well understood Reed-Muller codes on $\PP^1$. More recently C. Carvalho and V. G.L. Neumann \cite{CN} examined the case of rational surface scrolls $S(a_1,a_2)$ as subvarieties of $\PP^{a_1+a_2+1}$, which extends the result on $\PP^1\times\PP^1$, isomorphic to $S(1,1)$.

\medskip

In this paper we establish the parameters of Goppa codes corresponding to complete linear systems on minimal Hirzebruch surfaces $\mathcal{H}_\eta$, a family of projective toric surfaces indexed by $\eta \in \N$. This framework expands preceding works while taking advantage of both toric and projective features.

Regarding toric codes, we extend the evaluation map on the whole toric variety. This is analogous to the extension of affine Reed-Muller codes by projective ones introduced by G. Lachaud \cite{Lachaud}, since we also evaluate at "points at infinity". In other words toric codes on Hirzebruch surfaces can be obtained by puncturing the codes studied here at the $4q$ points lying on the $4$ torus-invariant divisors, that have at least one zero coordinate. As in the Reed-Muller case, through the extension process, the length turns to grow about twice as much as the minimal distance, as proved in Section \ref{punc}. 

Respecting the projective codes cited above, it turns out that rational surface scrolls are the range of some projective embeddings of a Hirzeburch surface, $\mathcal{H}_0$ for $\PP^1 \times \PP^1$ and  $\mathcal{H}_{a_1-a_2}$ for $S(a_1,a_2)$. However no embedding of the Hirzebruch surface into a projective space is required for our study and the Cox ring replaces the usual $\F_q[X_0,\dots,X_r]$ used in the projective context. Moreover, the embedded point of view forces to only evaluate polynomials of the Cox ring that are pullbacks of homogeneous polynomials of $\F_q[X_0,X_1,\dots,X_r]$ under this embedding. No such constraint appears using the Cox ring and polynomials of any bidegree can be examined.

\medskip

Whereas coding theorists consider evaluation codes with an injective evaluation map, C. Carvalho and V. G.L. Neumann (loc. cit.) extensively studied codes associated to a non necessarily injective evaluation map. In the present work no assumption of injectivity is needed. In particular, the computation of the dimension of the code does not follow from Riemann-Roch theorem. For a given degree, this grants us a wider range of possible sizes for the alphabet, including the small ones.

\medskip

Our study focuses on minimal Hirzeburch surfaces, putting aside $\mathcal{H}_1$, the blown-up of $\PP^2$ at a point. Although most techniques can be used to tackle this case, some key arguments fail, especially when estimating the minimal distance.

\medskip

The linear code $C_\eta(\delta_T,\delta_X)$ is defined as the evaluation code on $\F_q$-rational points of $\mathcal{H}_\eta$ of the set $R(\delta_T,\delta_X)$ of homogeneous polynomials of bidegree $(\delta_T,\delta_X)$, defined in Section \ref{secdebut}. The evaluation is naively not well-defined for a polynomial but a meaningful definition \textit{à la} Lachaud \cite{Lachaud} is given in Paragraph \ref{SensEval}.

Here the parameters of the code $C_\eta(\delta_T,\delta_X)$ are displayed as nice combinatoric quantities, from which quite intricate but explicit formulae can be deduced in Propositions \ref{dimcode} and \ref{distance}. The rephrasing of the problem in combinatorial terms is already a key feature in Hansen's \cite{HanToric} and Carvalho and Neumann's works \cite{CN} that is readjusted here to fit a wider range of codes.

A natural way to handle the dimension of these codes is to calculate the number of classes under the equivalence relation $\equiv$ on the set $R(\delta_T,\delta_X)$ that identifies two polynomials if they have the same evaluation on every $\F_q$-rational point of the Hirzebruch surface. Our strategy is to first restrict the equivalence relation $\equiv$ on the set of monomials $\Mon(\delta_T,\delta_X)$ of $R(\delta_T,\delta_X)$ and a handy characterization for two monomials to be equivalent is given. 

In most cases comprehending the equivalence relation over monomials is enough to compute the dimension. We have to distinguish a particular case:
\[\begin{array}{ccccccc}\tag{H}\label{H}
\eta \geq 2, &\delta_T < 0,& \eta \: | \: \delta_T, & q \leq \delta_X+\frac{\delta_T}{\eta}.
\end{array}\] 

\begin{theo}\label{DIM} The dimension of the code $C_\eta(\delta_T,\delta_X)$ satisfies
\[ \dim C_\eta(\delta_T,\delta_X) = \#  \left(\Mon(\delta_T,\delta_X)/\equiv\right) - \epsilon,\]
where $\epsilon$ is equal to $1$ if the couple $(\delta_T,\delta_X)$ satisfies (\ref{H}) and $0$ otherwise.
\end{theo}

This quantity depends on the parameter $\eta$, the bidegree $(\delta_T,\delta_X)$ and the size $q$ of the finite field.  

\smallskip

As for the dimension, the first step to determine the minimum distance is to bound it by below with a quantity that only depends on monomials. Again the strategy is similar to Carvalho and Neumann's one \cite{CN} but, even though they mentioned Gr\"obner bases, they did not fully benefit from the potential of the tools provided by Gr\"obner bases theory.
Indeed linear codes naturally involve linear algebra but the problem can be considered from a commutative algebra perspective. On this purpose, we consider the homogeneous vanishing ideal $\mathcal{I}$ of the subvariety constituted by the $\F_q$-rational points. A good understanding of a Gr\"obner basis of $\mathcal{I}$, through Section \ref{secgrob}, shortens the proof of the following theorem.

\begin{theo}\label{DIST}Let us fix $(\epsilon_T,\epsilon_X) \in \N^2$ such that $\epsilon_T, \: \epsilon_X \geq q$. The minimum distance $d_\eta(\delta_T,\delta_X)$ satisfies
\[d_\eta(\delta_T,\delta_X)  \geq \min_{M \in \Delta^*(\delta_T,\delta_X)} \# \Delta^*(\epsilon_T,\epsilon_X)_M\]
where $\Delta^*(\epsilon_T,\epsilon_X)_M$ is defined in Notation \ref{DeltaM}. It is an equality for $\epsilon_T=\delta_T+\eta\delta_X+q$ and $\epsilon_X=\delta_X+q$.
\end{theo}

The cardinality of $\Delta^*(\delta_T,\delta_X)$ depends on the parameter $\eta$, the bidegree $(\delta_T,\delta_X)$ and the size $q$ of the finite field.

\smallskip

The pullback of homogeneous polynomials of degree $\delta_X$ on $S(a_1,a_2) \subset \PP^r$ studied by C. Carvalho and V. G.L. Neumann are polynomials of bidegree $(a_2\delta_X,\delta_X)$ on $\mathcal{H}_{a_1-a_2}$. C. Carvalho and V. G.L. Neumann gave a lower bound of the minimum distance that we prove to be reached since it matches the parameters we establish here. The parameters also coincide with the one given by A. Couvreur and I. Duursma \cite{CD} in the case of the biprojective space $\PP^1 \times \PP^1$, isomorphic to Hirzebruch surface $\mathcal{H}_0$. 

It is worth pointing out that the codes $C_\eta(\delta_T,\delta_X)$ with $\delta_T$ negative have never been studied until now. Although this case is intricate when the parameter $\eta$ divides $\delta_T$ and the situation (\ref{H}) occurs, it brings the ideal $\mathcal{I}$ to light as an example of a non binomial ideal on the toric variety $\mathcal{H}_\eta$.

The last section highlights an interesting feature of these codes which leads to a good puncturing. It results codes of length $q(q+1)$ but with identical dimension and minimum distance.

\smallskip

We emphasize that the lower bound of the minimum distance in this paper does not result from upper bound of the number of rational point of embedded curves but from purely algebraic and combinatoric considerations. This approach, already highlighted by Couvreur and Duursma \cite{CD}, stands out from the general idea that one would estimate the parameters of an evaluation code on a variety $X$ though the knowledge of features of $X$, like some cohomology groups for the dimension or the number of rational points of subvarieties of $X$ for the minimum distance. It also offers the great perspective of solving geometric problems thanks to coding theory results. Moreover, the non injectivity of the evaluation map means that there exists a \emph{filling curve}, i.e. a curve that contains every $\F_q$-rational point of $\mathcal{H}_\eta$. From a number theoretical point of view, the minimum distance provides an upper bound of the number of $\F_q$-rational points of a non filling curve, regardless of its geometry and its smoothness, even if there exist some filling curves.

\section{Defining evaluation codes on Hirzebruch surfaces}\label{secdebut}

\subsection{Hirzebruch surfaces}

We gather here some results about Hirzebruch surfaces over a field $k$, given in \cite{CoxToric} for instance.

Let $\eta$ be a non negative integer. The \emph{Hirzebruch surface} $\mathcal{H}_\eta$ can be considered from different points of view.

\medskip

On one hand, the Hirzebruch surface $\mathcal{H}_\eta$ is the toric variety corresponding to the fan $\Sigma_\eta$ (see Figure \ref{f1}).

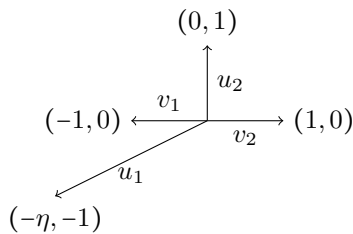
\begin{figure}[h]
\centering
\begin{tikzpicture}
\draw [->] (0,0) -- (0,1) node[above]{$(0,1)$} node[midway,right]{$u_2$};;
\draw [->] (0,0) -- (1,0) node[right]{$(1,0)$} node[midway,below]{$v_2$} ;
\draw [->] (0,0) -- (-1,0) node[left]{$(-1,0)$} node[midway,above]{$v_1$} ;
\draw [->] (0,0) -- (-2,-1) node[below]{$(-\eta,-1)$}node[midway,below]{$u_1$};
\end{tikzpicture}
\caption{Fan $\Sigma_\eta$}\label{f1}
\end{figure}

The fan $\Sigma_\eta$ being a refining of the one of $\PP^1$, it yields a ruling $\mathcal{H}_\eta \rightarrow \PP^1$ of fiber $\mathcal{F} \simeq \PP^1$ and section $\sigma$. The torus-invariant divisors $D_1$, $D_2$, $E_1$ and $E_2$ corresponding to the rays spanned respectively by $v_1$, $v_2$, $u_1$, $u_2$ generate the Picard group of $\mathcal{H}_\eta$, described in the following proposition.

\begin{prop}
The Picard group of the Hirzebruch surface $\mathcal{H}_\eta$ is the free Abelian group
\[\Pic  \mathcal{H}_\eta= \Z \mathcal{F} + \Z \sigma\]
where
\begin{equation}\label{gen}
\mathcal{F} = E_1 \sim E_2 \text{ and }
\sigma = D_2 \sim D_1 + \eta E_1.
\end{equation}
We have the following intersection matrix.
\[\begin{array}{c|cc}& \mathcal{F} & \sigma \\
\hline
\mathcal{F} & 0 & 1 \\
\sigma & 1 & \eta \end{array}\]
\end{prop}

As a simplicial toric variety, the surface $\mathcal{H}_\eta$ considered over $k$ carries a \emph{Cox ring} $R=k[T_1,T_2,X_1,X_2]$. Each monomial $M=T_1^{c_1} T_2 ^{c_2} X_1^{d_1} X_2^{d_2}$ of $R$ is associated to a torus-invariant divisor 
\begin{equation}\label{DIV}D_M=d_1 D_1 +d_2 D_2 + c_1 E_1 + c_2 E_2.\end{equation}
The \emph{degree} of the monomial $M$ is defined as the Picard class of the divisor $D_M$. The couple of coordinates $(\delta_T,\delta_X)$ of $D_M$ in the basis $(\mathcal{F},\sigma)$ is called the \emph{bidegree} of $M$ and denoted by $\bideg(M)$. By (\ref{gen}) and (\ref{DIV}),
\begin{equation}\label{systdelta}
\left\{\begin{array}{rl} \delta_T&=c_1+c_2-\eta d_1,\\
\delta_X&=d_1+d_2.\end{array}\right.
\end{equation}
It is convenient to set 
\[\delta=\delta_T+\eta \delta_X.\]

This gives the $\Z^2$-grading on $R$
\[R=\bigoplus_{(\delta_T,\delta_X) \in \Z^2} R(\delta_T,\delta_X)\]
where $R(\delta_T,\delta_X) \simeq H^0(\mathcal{H}_\eta, \mathcal{O}_{\mathcal{H}_\eta}(\delta_T \mathcal{F} +\delta_X \sigma))$ is the $k$-module of homogeneous polynomials of bidegree $(\delta_T,\delta_X)\in \Z^2$. Note that the $\F_q$-module $R(\delta_T,\delta_X)$ is non zero if and only if $\delta_X \in \N$ and $\delta \in \N$.

\medskip

On the other hand, the Hirzebruch surface can be displayed as a geometric quotient of an affine variety under the action of an algebraic group (\cite{CoxToric} Theorem $5.1.11$ ). This description is given for instance by M. Reid \cite{Reid}.

Let us define an action of the product of multiplicative groups $\Gb_m \times \Gb_m$ over $\left(\A^2 \setminus \{(0,0)\}\right) \times \left(\A^2 \setminus \{(0,0)\}\right)$: write $(t_1,t_2)$ for the first  coordinates on $\A^2$, $(x_1,x_2)$ on the second coordinates on $\A^2$ and $(\lambda,\mu)$ for elements of $\Gb_m \times \Gb_m$. The action is
given as follows:
\[(\lambda,\mu) \cdot(t_1,t_2,x_1,x_2)=(\lambda t_1 , \lambda t_2, \mu\lambda^{-\eta} x_1,\mu x_2).\]
Then the Hirzebruch surface $\mathcal{H}_\eta$ is isomorphic to the geometric quotient 
\[\left(\A^2 \setminus \{(0,0)\}\right) \times \left(\A^2 \setminus \{(0,0)\}\right) / \Gb_m^2.\]

This description enables us to describe a point of $\mathcal{H}_\eta$ by its homogeneous coordinates $(t_1,t_2,x_1,x_2)$.

\bigskip

In this paper, we focus only on minimal Hirzebruch surfaces. A surface is minimal if it contains no $-1$ curve. We recall the following well-known result about minimal Hirzebruch surface.

\begin{theorem}[\cite{LP}]
The Hirzeburch surface $\mathcal{H}_\eta$ is minimal if and only if $\eta \neq 1$.
\end{theorem}

\subsection{Evaluation map}\label{SensEval}

We consider now the case $k=\F_q$, $q$ being a power of a prime integer.

\medskip

From the ruling $\mathcal{H}_\eta \rightarrow \PP^1$, the number of $\F_q$ rational points of the Hirzebruch surface $\mathcal{H}_\eta$ is
\[N=\# \mathcal{H}_\eta (\F_q)=(q+1)^2.\]

Let $(\delta_T,\delta_X) \in\Z\times\N$ such that $\delta \geq 0$. Given a polynomial $F \in R(\delta_T,\delta_X)$ and a point $P$ of  $\mathcal{H}_\eta$, \emph{the evaluation of $F$ at $P$} is defined by $F(P) = F(t_1,t_2,x_1,x_2)$, where $(t_1,t_2,x_1,x_2)$ is the only tuple that belongs to the orbit of $P$ under the action of $\Gb_m^2$ and has one of these forms:
\begin{itemize}
\item $(1,a,1,b)$ with $a, \:b \in \F_q$,
\item $(0,1,1,b)$ with $b \in \F_q$,
\item $(1,a,0,1)$ with $a\in \F_q$,
\item $(0,1,0,1)$.
\end{itemize}

The \emph{evaluation code} $C_\eta(\delta_T,\delta_X)$ is defined as the image of the evaluation map
\begin{equation}\label{evaluation}\ev_{(\delta_T,\delta_X)}: \left\{ \begin{array}{rcl} R(\delta_T,\delta_X) & \rightarrow & \F_q^N \\
F & \mapsto & (F(P))_{P \in \mathcal{H}\eta(\F_q)}.
\end{array}\right. \end{equation}
Note that this code is Hamming equivalent to the Goppa code $C(\mathcal{O}_{\mathcal{H}_\eta}(\delta_T \mathcal{F} + \delta_X \sigma),\mathcal{H}_\eta(\F_q))$, as defined by Hansen \cite{Han}. The \emph{weight} $\omega(c)$ of a codeword $c\in C_\eta(\delta_T,\delta_X)$ is the number of non-zero coordinates. The minimum weight among all the non-zero codewords is called the \emph{minimum distance} of the code $C_\eta(\delta_T,\delta_X)$ and is denoted by $d_\eta(\delta_T,\delta_X)$.

\section{Dimension of the evaluation code $C_\eta(\delta_T,\delta_X)$ on the Hirzebruch surface $\mathcal{H}_\eta$}\label{secdim}

Let us consider $\eta \geq 0$ and $(\delta_T,\delta_X) \in \Z \times \N$ such that $\delta=\delta_T+\eta \delta_X \geq 0$.
\begin{notation}
The kernel of the map $\ev_{(\delta_T,\delta_X)}$ is denoted by $\mathcal{I}(\delta_T,\delta_X)$.
\end{notation}

From the classical isomorphism
\[ C_\eta(\delta_T,\delta_X) \simeq \faktor{R(\delta_T,\delta_X)}{\mathcal{I}(\delta_T,\delta_X)},\]
the dimension of the evaluation code $C_\eta(\delta_T,\delta_X)$ equals the dimension of any complementary vector space of $\mathcal{I}(\delta_T,\delta_X)$ in $R(\delta_T,\delta_X)$. This is tantamount to compute the range of a well-chosen projection map on $R(\delta_T,\delta_X)$ along $\mathcal{I}(\delta_T,\delta_X)$.

\subsection{Focus on monomials}

The aim of this section is to display a projection map, denoted by $\pi_{(\delta_T,\delta_X)}$, that would have the good property of mapping a monomial onto a monomial. The existence of such a projection is not true in full generality: given a vector subspace $W$ of a vector space $V$ and a basis $\mathcal{B}$ of $V$, it is not always possible to find a basis of $W$ composed of difference of elements of $\mathcal{B}$ and a complementary space of $W$ which basis is a subset of $\mathcal{B}$. This will be possible here except if (\ref{H}) holds.

With this goal in mind, our strategy is to focus first on monomials of $R(\delta_T,\delta_X)$. Let us define the following equivalence relation on the set of monomials of $R(\delta_T,\delta_X)$.
\begin{definition}
Let us define a binary relation $\equiv$ on the set $\Mon(\delta_T,\delta_X)$ of monomials of $R(\delta_T,\delta_X)$. Let $M_1, \: M_2 \in \Mon(\delta_T,\delta_X)$. We note $M_1 \equiv M_2$ if they have the same evaluation at every $\F_q$-rational point of $\mathcal{H}_\eta$, i.e. 
\[M_1 \equiv M_2 \: \Leftrightarrow \: \ev_{(\delta_T,\delta_X)}(M_1)=\ev_{(\delta_T,\delta_X)}(M_2) \: \Leftrightarrow \: M_1-M_2 \in \mathcal{I}(\delta_T,\delta_X).\]
\end{definition}
This section is intended to prove that, even if this equivalence relation can be defined over all $R(\delta_T,\delta_X)$, the number of equivalence classes when considering all polynomials is the same as when regarding only monomials, unless (\ref{H}) holds. This section thus goals to prove Theorem \ref{DIM}, stated in the introduction.

\subsection{Combinatorial point of view of the equivalence relation on monomials}

Throughout this article, the set $R(\delta_T,\delta_X)$ are pictured as a polygon is $\N\times \N$ of coordinates $(d_2,c_2)$. This point of view, inherited directly from the toric structure, is common in the study of toric codes (\cite{HanToric}, \cite{Joyner}, \cite{Ruano}, \cite{Little}, \cite{SS}). It will be useful to handle the computation of the dimension and the minimum distance as a combinatorial problem.

\begin{definition}\label{polygone}
Let $(\delta_T,\delta_X) \in \Z \times \N$. Let us define the polygon
\[P_D=\{ (a,b) \in \R^2 \: | \: a \geq 0, \: b \geq 0, \: a \leq \delta_X \text{ and } \eta a + b \leq \delta\}\]
associated to the divisor $D=\delta E_1 + \delta_X D_1 \sim \delta_T \mathcal{F} + \delta_X \sigma$ and
\[\mathcal{P}(\delta_T,\delta_X)=P_D \cap \Z^2.\]
\end{definition}

Being intersection of $\Z^2$ with half planes, it is easily seen that $\mathcal{P}(\delta_T,\delta_X)$ is the set of lattice points of the polygon $P_D$, which vertices are
\begin{itemize}
\item $(0,0),(\delta_X,0),(\delta_X,\delta_T), (0,\delta)$ if $\delta_T > 0$,
\item $(0,0),(\frac{\delta}{\eta},0),(0,\delta)$ if $\delta_T < 0$ and $\eta > 0$ or $\delta_T=0$.
\end{itemize}
Note that $P_D$ is a lattice polygone except if $\delta_T <0$ and $\eta$ does not divide $\delta_T$.

\begin{notation}\label{notA}
Let us set
\[A=A(\eta,\delta_T,\delta_X)=\min\left(\delta_X,\frac{\delta}{\eta}\right)=\left\{\begin{array}{cl}\delta_X & \text{if } \delta_T \geq 0, \\
\frac{\delta}{\eta}=\delta_X+\frac{\delta_T}{\eta} &\text{otherwise,}\end{array}\right.\]
the $x$-coordinate of the right-most vertices of the polygon $P_D$.
\end{notation}

Let us highlight that  $A$ is not necessary an integer if $\delta_T < 0$. Thus it does not always appear as the first coordinate of an element of $\mathcal{P}(\delta_T,\delta_X)$. It is the case if and only if $\eta \: | \: \delta_T$. If so, the only element of $\mathcal{P}(\delta_T,\delta_X)$ such that $A$ is its first coordinate is $(A,0)$.

We thus observe that 
\begin{equation}\label{Pexplicite}\mathcal{P}(\delta_T,\delta_X)=\{(a,b) \in \N^2 \: | \: a \in \lsem 0 , \lfloor A \rfloor \rsem \text{ and } b \in \lsem 0,\delta_T+\eta(\delta_X-a)\rsem\}.\end{equation}

\begin{figure}[h!]

\begin{subfigure}[b]{0.3\textwidth}
  \centering
  \begin{tikzpicture}[scale=0.4]

    \clip (-3,-2) rectangle (6cm,10cm); 

    \draw[style=help lines,dashed,dashed,opacity=0.5] (-3,-10) grid[step=1cm] (12,14); 

 \draw [thick,-latex] (0,0) -- (0,9) node [left,font=\small] {$c_2$};
    \draw [thick,-latex] (0,0) -- (5,0) node [below right,font=\small] {$d_2$};
        
    \draw [very thick,blue] (0,0) -- (0,7)  node [left, font=\small] {$\delta=\delta_T$};
    \draw [very thick,blue] (0,0) -- (4,0) node [below, font=\small] {$A=\delta_X$};
    \draw [very thick,blue] (4,0) -- (4,7) ;
   \draw [very thick,blue] (0,7) -- (4,7) ;
   \end{tikzpicture}
\caption{$\eta=0$ \\
e.g. $\mathcal{P}(7,4)$ }
\end{subfigure}
\begin{subfigure}[b]{0.3\textwidth}
  \centering
  \begin{tikzpicture}[scale=0.4]

     \clip (-2,-2) rectangle (7cm,10cm); 

    \draw[style=help lines,dashed,dashed,opacity=0.5] (-4,-10) grid[step=1cm] (14,14); 

 \draw [thick,-latex] (0,0) -- (0,9)  node [left,font=\small] {$c_2$};
    \draw [thick,-latex] (0,0) -- (5,0) node [below right,font=\small] {$d_2$};
        
    \draw [very thick,blue] (0,0) -- (0,8) node [left, font=\small] {$\delta$};
    \draw [very thick,blue] (0,0) -- (3,0) node [below, font=\small] {$A=\delta_X$} ;
       \draw [very thick,blue] (3,2) -- (0,8);
   \draw [very thick,blue] (3,2) -- (3,0) ;
   \end{tikzpicture}
\caption{$\eta >0$, $\delta_T > 0$ \\
e.g. $\mathcal{P}(2,3)$ in $\mathcal{H}_2$}
\end{subfigure}
\begin{subfigure}[b]{0.3\textwidth}
  \centering
  \begin{tikzpicture}[scale=0.4]

      \clip (-2,-2) rectangle (8cm,10cm); 

    \draw[style=help lines,dashed,dashed,opacity=0.5] (-4,-10) grid[step=1cm] (14,14); 

 \draw [thick,-latex] (0,0) -- (0,9) node [left,font=\small] {$c_2$};
    \draw [thick,-latex] (0,0) -- (6,0) node [below right,font=\small] {$d_2$};
        
    \draw [very thick,blue] (0,0) -- (0,8) node [left, font=\small] {$\delta$};
    \draw [very thick,blue] (0,0) -- (4,0) node [below, font=\small] {$A<\delta_X$}  ;
   \draw [very thick,blue] (0,8) -- (4,0) ;
   \end{tikzpicture}
\caption{$\eta >0$, $\delta_T \leq 0$ \\
e.g. $\mathcal{P}(-2,5)$ in $\mathcal{H}_2$}
\end{subfigure}
\caption{Different shapes of the polygon $\mathcal{P}(\delta_T,\delta_X)$}\label{f2}
\end{figure}
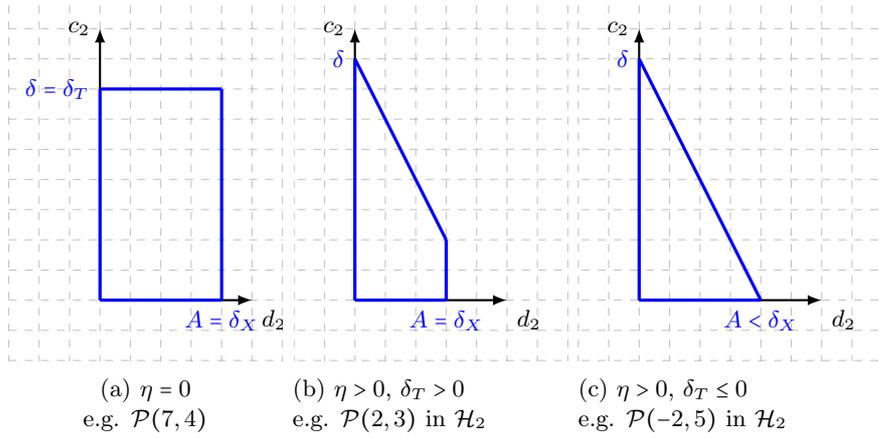

\begin{example}
Figure \ref{f2} gives the three examples of possible shapes of the polygon $\mathcal{P}(\delta_T,\delta_X)$. The first one is the case $\eta=0$, and the last two ones correspond to $\eta > 0$ and depend on the sign of $\delta_T$, which determines the shape of $P_D$. All proofs of explicit formulae in Propositions \ref{dimcode} and \ref{distance} distinguish these cases.
\end{example}

Thanks to (\ref{systdelta}), a monomial of $R(\delta_T,\delta_X)$ is entirely determined by the couple $(d_2,c_2)$. Then each element of $\mathcal{P}(\delta_T,\delta_X)$ corresponds to a unique monomial. More accurately, for any couple $(d_2,c_2) \in \mathcal{P}(\delta_T,\delta_X)$, we define the monomial
\begin{equation}\label{Mdc}M(d_2,c_2)=T_1^{\delta_T+\eta(\delta_X-d_2)-c_2} T_2^{c_2} X_1^{\delta_X-d_2} X_2^{d_2} \in \Mon(\delta_T,\delta_X).\end{equation}

\begin{definition}\label{pont}
The equivalence relation $\equiv$ on  $\Mon(\delta_T,\delta_X)$ and the bijection
\begin{equation}\label{bijection}\left\{\begin{array}{rcl}
\mathcal{P}(\delta_T,\delta_X) & \rightarrow &  \Mon(\delta_T,\delta_X)\\
(d_2,c_2) & \mapsto & M(d_2,c_2) \end{array}\right.\end{equation}
endow $\mathcal{P}(\delta_T,\delta_X)$ with a equivalence relation, also denoted by $\equiv$, such that
\[(d_2,c_2) \equiv (d'_2,c'_2) \: \Leftrightarrow \: M(d_2,c_2) \equiv M(d'_2,c'_2).\]
\end{definition}

\begin{prop}\label{traduc}
Let two couples $(d_2,c_2)$ and $(d'_2,c'_2)$ be in $\mathcal{P}(\delta_T,\delta_X)$ and let us write
\[M=M(d_2,c_2)=T_1^{c_1}T_2^{c_2}X_1^{d_1}X_2^{d_2} \text{ and } M'=M(d'_2,c'_2)=T_1^{c'_1}T_2^{c'_2}X_1^{d'_1}X_2^{d'_2}.\]
Then $(d_2,c_2)\equiv(d'_2,c'_2)$ if and only if
\begin{align}
\tag{C1} \label{D1}  q-1 \: &| \: d_i-d'_i, \\
\tag{C2} \label{D2} q-1 \: &| \: c_j-c'_j, \\
\tag{C3} \label{D3} d_i = 0 \: &\Leftrightarrow d'_i=0, \\
\tag{C4} \label{D4} c_j = 0 \: &\Leftrightarrow c'_j=0.
\end{align}
\end{prop}

\begin{proof}The conditions (\ref{D1}), (\ref{D2}), (\ref{D3}) and (\ref{D4}) clearly imply that $M(d_2,c_2) \equiv M(d'_2,c'_2)$, hence $(d_2,c_2)\equiv(d'_2,c'_2)$. To prove the converse, assume that $M \equiv M'$ and write
\[M=M(d_2,c_2)=T_1^{c_1}T_2^{c_2}X_1^{d_1}X_2^{d_2} \text{ and } M'=M(d'_2,c'_2)=T_1^{c'_1}T_2^{c'_2}X_1^{d'_1}X_n^{d'_2}.\]

Let $x \in \F_q$. Then $M(1,x,1,1)=M'(1,x,1,1)$, which means
$x^{c_2}=x^{c'_2}$. But this equality is true for any element $x$ of $\F_q$ if and only if $(T_2^q-T_2) \: | \: (T_2^{c_2}-T_2^{c'_2})$. This is equivalent to $c_2=c'_2=0$ or $c_2c'_2 \neq 0$ and $T_2^{q-1} - 1 \: | \: T_2^{c'_2-1}(T_2^{c_2-c'_2} - 1)$, which proves  (\ref{D2}) and (\ref{D4}) for $i=2$.

Repeating this argument evaluating at $(1,1,1,x)$ for every $x \in \F_q$ gives $q-1 \: | \: d_2-d'_2$ and $d_2=0$ if and only if $d'_2=0$, i.e. (\ref{D1}) and (\ref{D3}) for $i=2$.

Moreover, we have $d_1+d_2=d'_1+d'_2=\delta_X$, which means that $q-1 \: | \: d_2-d'_2$ if and only if $q-1 \: | \: d_1-d'_1$. Evaluating at $(1,1,0,1)$ gives $0^{d_1}=0^{d'_1}$. Then $d_1=0$ if and only if $d'_1=0$. This proves (\ref{D1}) and (\ref{D3}) for $i=1$.

It remains the case of $c_1$ and $c'_1$. We have
\[c_1-c'_1=c'_2-c_2 -\eta (d'_1-d_1)\]
and $q-1$ divides $c_2-c'_2$ and $d'_1-d_1$. Then it also divides $c_1-c'_1$. Evaluating at $(0,1,1,1)$ yields like previously $c_1=0$ of and only if $c'_1=0$.
\end{proof}

\begin{remark}
The conditions of Lemma \ref{traduc} also can be written
\begin{align}
c_i=c'_i=0 &\text{ or } c_ic'_i \neq 0 \text{ and } q-1 \: | \: c'_i-c_i,\\
d_i=d'_i=0 &\text{ or } d_id'_i \neq 0 \text{ and } q-1 \: | \: d'_i-d_i.
\end{align}
Besides, the conditions involving $q$ are always satisfied for $q=2$.
\end{remark}

\begin{obs}\label{seul}
The conditions (\ref{D3}) and (\ref{D4}) mean that a point of $\mathcal{P}(\delta_T,\delta_X)$ lying on an edge of $P_D$ can be equivalent only with a point lying on the same edge. Therefore the equivalence class of a vertex of $P_D$ is a singleton.
\end{obs}

To prove that the number of equivalence classes equals the dimension of the code $C_\eta(\delta_T,\delta_X)$ as stated in Theorem \ref{DIM} (unless (\ref{H}) holds), we goal to choose a set $\mathcal{K}(\delta_T,\delta_X)$ of representatives of the equivalence classes of $\mathcal{P}(\delta_T,\delta_X)$ under the relation $\equiv$, which naturally gives a set of representatives $\Delta(\delta_T,\delta_X)$ for $\Mon(\delta_T,\delta_X)$ under the binary relation $\equiv$.

\begin{notation}\label{notrep}
Let $(\delta_T,\delta_X) \in \Z \times \N$ and $q \geq 2$. Let us set
\[\begin{aligned}\mathcal{A}_X&=\{ \alpha \in \N \: | \: 0 \leq \alpha \leq \min(\lfloor A \rfloor, q -1)\}\cup\{A\} \cap \N,\\
\mathcal{K}(\delta_T,\delta_X)&=\left\lbrace  (\alpha,\beta) \in \N^2 \: \left| \:
\begin{array}{c} \alpha \in \mathcal{A}_X \\
0 \leq \beta \leq \min(\delta -\eta\alpha,q)-1 \text{ or } \beta=\delta -\eta\alpha\\
\end{array}\right.\right\rbrace,\\
\Delta(\delta_T,\delta_X)&=\{ M(\alpha,\beta) \: | \: (\alpha,\beta) \in \mathcal{K}(\delta_T,\delta_X)\}.\end{aligned}\]
\end{notation}

Notice that $\mathcal{K}(\delta_T,\delta_X)$ is nothing but $\mathcal{P}(\delta_T,\delta_X)$ cut out by the set
\[\left(\{d_2 \leq q-1\} \cup \{d_2=A\} \right) \cap \left(\{c_2 \leq q-1\} \cup \{c_2=\delta-\eta d_2)\} \right).\]

\begin{example}\label{firstexample}
Let us set $\eta=2$ and $q=3$. Let us sort the monomials of $\mathcal{M}(-2,5)$, grouping the ones with the same image under $\ev_{(-2,5)}$, using Proposition \ref{traduc}.

Figure \ref{f4} represents the set $\mathcal{K}(-2,5)$. Note that for each couple $(d_2,c_2) \in\mathcal{K}(-2,5)$, there is exactly one of these groups that contains the monomial $M(d_2,c_2)$.

\[\begin{array}{c|c}
\text{Exponents }(c_1,c_2,d_1,d_2) \text{ of } T_1^{c_1}T_2^{c_2}X_1^{d_1}X_2^{d_2} & \text{Couple in }\mathcal{K}(-2,5)\\
\hline
(8,0,5,0)&(0,0)\\
(7,1,5,0) \sim (5,3,5,0) \sim (3,5,5,0) \sim (1,7,5,0)&(0,1)\\
(6,2,5,0) \sim (4,4,5,0) \sim (2,6,5,0)&(0,2)\\
(0,8,5,0)&(0,8)\\
(6,0,4,1) \sim (2,0,2,3)&(1,0)\\
(5,1,4,1) \sim (3,3,4,1) \sim (1,5,4,1) \sim (1,1,2,3) & (1,1)\\
(4,2,4,1) \sim (2,4,4,1) & (1,2)\\
(0,6,4,1) \sim (0,2,2,3) & (1,6)\\
(4,0,3,2) & (0,2)\\
(3,1,3,2) \sim (1,3,3,2) & (2,1)\\
(2,2,3,2) & (2,2)\\
(0,4,3,2) & (2,4)\\
(0,0,1,4) & (4,0)
\end{array}\]
\end{example}

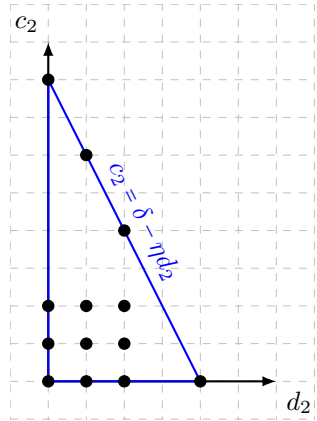
\begin{figure}[h!]
  \centering
  \begin{tikzpicture}[scale=0.5]

    \clip (-1,-1) rectangle (7cm,10cm); 

    \draw[style=help lines,dashed,dashed,opacity=0.5] (-1,-10) grid[step=1cm] (14,14); 

 \draw [thick,-latex] (0,0) -- (0,9) node [above left] {$c_2$};
    \draw [thick,-latex] (0,0) -- (6,0) node [below right] {$d_2$};
        
    \draw [thick,blue] (0,0) -- (0,8);
    \draw [thick,blue] (0,0) -- (4,0) ;
   \draw [thick,blue] (0,8) -- (4,0) node [midway,above,sloped] {$c_2=\delta-\eta d_2$} ;
   
 \foreach \x in {0,1,2}{                           
      \foreach \y in {0,1,2,8-2*\x}{                       
        \node[draw,circle,inner sep=1.5pt,fill,black] at (\x,\y) {}; 
      }
    }   
   \node[draw,circle,inner sep=1.5pt,fill,black] at (4,0) {};
  \end{tikzpicture}
\caption{Dots in $\mathcal{P}(-2,5)$ correspond to elements of $\mathcal{K}(-2,5)$.}\label{f4}
\end{figure}

Motivated by Example \ref{firstexample}, we give a map that displays $\mathcal{K}(\delta_T,\delta_X)$ as a set of representatives of $\mathcal{P}(\delta_T,\delta_X)$ under the equivalence relation $\equiv$.

\begin{definition}\label{proj}
Let us set the map $p_{(\delta_T,\delta_X)}:\mathcal{P}(\delta_T,\delta_X) \rightarrow \mathcal{P}(\delta_T,\delta_X)$ such that for every couple ${(d_2,c_2) \in \mathcal{P}(\delta_T,\delta_X)}$ its image $p_{(\delta_T,\delta_X)}(d_2,c_2)=(d'_2,c'_2)$ is defined as follows.
\begin{itemize}
\item If $d_2 =0$ or $d_2=A$, then $d'_2=d_2$,
\item Otherwise, we choose $d'_2 \equiv d_2 \mod q-1$ with $1 \leq d'_2 \leq q-1$.
\end{itemize}
and
\begin{itemize}
\item If $c_2 =0$, then $c'_2=0$,
\item If $c_2=\delta-\eta d_2$, then $c'_2=\delta-\eta d'_2$,
\item Otherwise, we choose $c'_2 \equiv c_2 \mod q-1$ with $1 \leq c'_2 \leq q-1$.
\end{itemize}
\end{definition}

\begin{prop}\label{equi} 
\begin{enumerate}
\item The map $p_{(\delta_T,\delta_X)}$ induces a bijection from the quotient set $\mathcal{P}(\delta_T,\delta_X)/\equiv $ to $\mathcal{K}(\delta_T,\delta_X)$.
\item The set $\mathcal{K}(\delta_T,\delta_X)$ is a set of representatives of $\mathcal{P}(\delta_T,\delta_X)$ under the equivalence relation $\equiv$.
\item The set $\Delta(\delta_T,\delta_X)$ is a set of representatives of $\Mon(\delta_T,\delta_X)$ under the equivalence relation $\equiv$.
\end{enumerate}
\end{prop}

\begin{proof}
First notice that elements of $\mathcal{K}(\delta_T,\delta_X)$ are invariant under $p_{(\delta_T,\delta_X)}$.

\smallskip

The inclusion $p_{(\delta_T,\delta_X)}(\mathcal{P}(\delta_T,\delta_X)) \subset \mathcal{K}(\delta_T,\delta_X)$ is clear by definitions of $\mathcal{K}(\delta_T,\delta_X)$ (Not. \ref{notrep}) and $p_{(\delta_T,\delta_X)}$ (Def. \ref{proj}). The equality follows from the invariance of $\mathcal{K}(\delta_T,\delta_X)$.

\smallskip

Last, we prove that $p_{(\delta_T,\delta_X)}(d_2,c_2) \equiv (d_2,c_2)$ for every couple $(d_2,c_2) \in \mathcal{P}(\delta_T,\delta_X)$.
Take a couple $(d_2,c_2) \in \mathcal{P}(\delta_T,\delta_X)$ and denote by $(d'_2,c'_2)$ its image under $p_{(\delta_T,\delta_X)}$. We have to prove that $(d_2,c_2)$ and $(d'_2,c_2)$ satisfy all the conditions of Proposition \ref{traduc}.

By definition of $p_{(\delta_T,\delta_X)}$, it is clear that conditions (\ref{D1}), (\ref{D2}), (\ref{D3}), as well as the the forward implication of (\ref{D4}), are true. It remains to prove that $c'_i=0 \: \Rightarrow c_i=0$ for $i \in \{1,2\}$.

Let us prove only the case $i=2$. So assume that $c'_2=0$. Then $c_2=0$ or $c_2=\delta-\eta d_2$. However,
\[c_2=\delta-\eta d_2 \: \Leftrightarrow \: c'_2=\delta-\eta d'_2=0 \: \Leftrightarrow \: d'_2=\frac{\delta}{\eta}.\]
This is only possible when $\delta_T \leq 0$ and then $d'_2=A$. By condition (\ref{D3}), this implies that $d_2=A$ and then $c_2=0$. This proves the first item.

The second assertion is a straightforward consequence of the first one.

Finally the third assertion yields from the definition of the equivalence relation $\equiv$ on $\mathcal{P}(\delta_T,\delta_X)$ via the bijection (\ref{bijection}). 
\end{proof}

\begin{corollary}\label{nbDelta}
The number of equivalence classes $\#\Delta(\delta_T,\delta_X)$ of $\Mon(\delta_T,\delta_X)$ under $\equiv$ is equal to the cardinality of $\mathcal{K}(\delta_T,\delta_X)$.
\end{corollary}

\begin{proof}
Its results from Definition \ref{pont} and Proposition \ref{equi}.
\end{proof}

\subsection{Proof of Theorem \ref{DIM}}

The main idea of the proof is to define an endomorphism on the basis of monomials $\Mon(\delta_T,\delta_X)$ by conjugation of $p_{(\delta_T,\delta_X)}$ by the bijection (\ref{bijection}) and prove it to be a projection along $\mathcal{I}(\delta_T,\delta_X)$ onto $\Span \Delta(\delta_T,\delta_X)$. However, when (\ref{H}) occurs, there is a non trivial linear combination of elements of $\Delta(\delta_T,\delta_X)$ lying in $\mathcal{I}(\delta_T,\delta_X)$, as pointed out in the following lemma.

\begin{lemma}\label{ajustement}
Let $(\delta_T,\delta_X) \in \Z^2$. Assume that $\eta \geq 2$, $\delta_T < 0$, $\eta \, | \, \delta_T$ and $q \leq \frac{\delta}{\eta}$, i.e. (\ref{H}) holds. Let us set $k \in \N$ and $r \in \lsem 1,q-1\rsem$ such that
\[A=\frac{\delta}{\eta}=k(q-1)+r.\]
The polynomial
\[\begin{aligned}F_0=&M(A,0) - M(r,0) + M(r,q-1)-M(r,\eta k (q-1))\\
&X_1^{-\frac{\delta_T}{\eta}}X_2^{\frac{\delta}{\eta}}-T_1^{\eta k(q-1)}X_1^{k(q-1)-\frac{\delta_T}{\eta}}X_2 ^r \\
&+ T_1^{(\eta k-1)(q-1)}T_2^{q-1}X_1^{k(q-1)-\frac{\delta_T}{\eta}}X_2^r - T_2^{\eta k (q-1)} X_1^{k(q-1)-\frac{\delta_T}{\eta}}X_2^r\end{aligned}\]
belongs to $\mathcal{I}(\delta_T,\delta_X)$.
\end{lemma}

\begin{proof}
Let us prove that the polynomial $F_0$ vanishes at every $\F_q$-rational of $\mathcal{H}_\eta$.

For any $a \in \F_q$, we have $F_0(1,a,0,1)=0$ and $F_0(0,1,0,1)=0$ since every polynomial in $R(\delta_T,\delta_X)$ is divisible by $X_1$ when $\delta_T <0$.

For $(a,b) \in \F_q^2$, $F_0(1,a,1,b)=b^{\frac{\delta}{\eta}} - b^r +a^{q-1} b ^r - a^{\eta k (q-1)} b^r=0$, as $q-1 \, | \, \frac{\delta}{\eta} -r \neq 0$.

For the same reason, $F_0(0,1,1,b)=b^{\delta_X+\frac{\delta_T}{\eta}} -0 + 0 -b^r=0$ for any $b \in \F_q$.
\end{proof}

The previous lemma displays a polynomial with 4 terms in the kernel when the couple $(\delta_T,\delta_X)$ satisfies (\ref{H}). We thus have to adjust the endomorphism in this case.

\begin{definition}\label{pi}
Let us set the linear map $\pi_{(\delta_T,\delta_X)}:R(\delta_T,\delta_X) \rightarrow R(\delta_T,\delta_X)$ such that for every $(d_2,c_2) \in P(\delta_T,\delta_X)$,
\[\pi_{(\delta_T,\delta_X)}(M(d_2,c_2))=M(p_{(\delta_T,\delta_X)}(d_2,c_2))\]
except for $(d_2,c_2)=\left(\frac{\delta}{\eta},0\right)$ when the couple $(\delta_T,\delta_X)$ satisfies (\ref{H}). In this case, set $(r,k)$ is the unique couple of integers such that $\frac{\delta}{\eta}= k(q-1)+ r$ with $r \in \lsem 1 ,q-1\rsem$ and
\[\pi_{(\delta_T,\delta_X)}\left(M\left(\frac{\delta}{\eta},0\right)\right)= M(r,0)+M(r,\eta k(q-1))-M(r,q-1).\]
\end{definition}

\begin{remark}\label{ok}
The monomials $M(r,0)$, $M(r,\eta k(q-1))$ and $M(r,q-1)$, that appears in the definition above, belong to $\Delta(\delta_T,\delta_X)$.
\end{remark}

\begin{notation}\label{etoile}
Let $(\delta_T,\delta_X) \in \Z  \times \N$ such that $\delta \geq 0$. If (\ref{H}) holds, we set
\[\begin{aligned}\mathcal{K}^*(\delta_T,\delta_X)&=\mathcal{K}(\delta_T,\delta_X) \setminus \left\{\left(\frac{\delta}{\eta},0\right)\right\}\\
&=
\left\lbrace  (\alpha,\beta) \in \N^2 \: \left| \:
\begin{array}{c} \alpha \in \lsem 0,q-1\rsem\\
0 \leq \beta \leq \min(\delta -\eta\alpha,q)-1 \text{ or } \beta=\delta -\eta\alpha\\
\end{array}\right.\right \rbrace ,\end{aligned}\]
and
\[\Delta^*(\delta_T,\delta_X)=\{ M(\alpha,\beta) \: | \: (\alpha,\beta) \in \mathcal{K}^*(\delta_T,\delta_X)\}.\]
Otherwise, we set
\[\mathcal{K}^*(\delta_T,\delta_X)=\mathcal{K}(\delta_T,\delta_X) \text{ and } \Delta^*(\delta_T,\delta_X)=\Delta(\delta_T,\delta_X).\]
\end{notation}

\begin{lemma}\label{zerocombi}
The only zero linear combination of elements of $\Delta^*(\delta_T,\delta_X)$ that belongs to $\mathcal{I}(\delta_T,\delta_X)$ is the trivial one.
\end{lemma}

\begin{proof}

Let us assume that a linear combination of elements of $\Delta^*(\delta_T,\delta_X)$
\[H=\sum_{(\alpha,\beta)\in\mathcal{K}^*(\delta_T,\delta_X)}\lambda_{\alpha,\beta} M(\alpha,\beta)\]
satisfies $\ev_{(\delta_T,\delta_X)} \left(H\right)=0$.

On one side, $H(1,0,1,0)=\lambda_{0,0}$, $H(1,0,0,1)=\lambda_{\delta_X,0}$,   $H(0,1,0,1)=\lambda_{\delta_X,\delta_T}$ and $H(0,1,1,0)=\lambda_{0,\delta}$. Then $\lambda_{0,0}=\lambda_{\delta_X,0}=\lambda_{\delta_X,\delta_T}=\lambda_{0,\delta}=0$. On the other side, evaluating at $(1,a,1,0)$ for any $a \in \F_q$ gives
\[H(1,a,1,0)=\sum_{\beta=1}^{\min(q-1,\delta-1)} \lambda_{0,\beta} \: a^\beta=0.\]
Then the polynomial
\[\sum_{\beta=1}^{\min(q-1,\delta-1)} \lambda_{0,\beta} \: X^\beta\]
of degree lesser than $(q-1)$ has $q$ zeros. This implies that $\lambda_{0,\beta}=0$ for any $\beta$ such that $(0,\beta) \in \mathcal{K}^*(\delta_T,\delta_X)$. Evaluating at $(1,a,0,1)$, we can deduce that $\lambda_{\delta_X,\beta}=0$ for any $\beta$ such that $(\delta_X,\beta) \in \mathcal{K}^*(\delta_T,\delta_X)$.

To evaluate at $(1,0,1,a)$, two cases are distinguished.
\begin{itemize}
\item If $\delta_T \geq 0$, 
\[H=(1,0,1,a)=\sum_{\alpha = 1}^{\min(\delta_X,q)-1)} \lambda_{\alpha,0} \: a^\alpha=0,\]
which implies with the same argument that $\lambda_{\alpha,0}=0$ for every $\alpha$ such that $(\alpha,B(\alpha)) \in \mathcal{K}^*(\delta_T,\delta_X)$.
\item If $\delta_T <0$,
\[H=(1,0,1,a)=\sum_{\alpha = 1}^{\min(\lfloor A \rfloor,q-1)} \lambda_{\alpha,0} \: a^\alpha=0\]
and we can repeat the same argument than before.
\end{itemize}

Similarly, by evaluating at $(0,1,1,a)$, we have $\lambda_{\alpha,B(\alpha)}=0$ for any $\alpha$ such that $(\alpha,B(\alpha)) \in \mathcal{K}^*(\delta_T,\delta_X)$.

For any $a, \: b \in \F_q$, we then have
\[H(1,a,1,b)=\sum_{\alpha=1}^{\min(q-1,\delta_X-1)} \left( \sum_{\beta=1}^{\min(q-1,B(\alpha)-1)}\lambda_{\alpha,\beta \:} a^{\beta}\right) b^\alpha=0\]
which implies that for any $a \in \F_q$, the polynomial 
\[\sum_{\alpha=1}^{\min(q-1,\delta_X-1)} \left( \sum_{\beta=1}^{\min(q-1,B(\alpha)-1)}\lambda_{\alpha,\beta} \: a^{\beta}\right) X^\alpha\]
of degree lesser than $(q-1)$ has $q$ zeros and, thus, is zero. By the same argument on each coefficient as polynomials of variable $a$, we then have proved that the linear combination $H$ is zero.
\end{proof}

Theorem \ref{DIM} follows from the following proposition. 

\begin{prop}\label{span}
The linear map $\pi_{(\delta_T,\delta_X)}$ is the projection along $\mathcal{I}(\delta_T,\delta_X)$ onto $\Span \Delta^*(\delta_T,\delta_X)$. Moreover the elements of $\Delta^*(\delta_T,\delta_X)$ are linearly independent.
\end{prop}

\begin{proof}
By construction of $\Delta^*(\delta_X,\delta_T)$, the definition of $\pi_{(\delta_T,\delta_X)}$ and Remark \ref{ok}, it is clear that $\range\pi_{(\delta_T,\delta_X)} \subset \Span \Delta^*(\delta_T,\delta_X)$. Also, by Proposition \ref{equi} and the bijection (\ref{bijection}), any monomial of $\Delta^*(\delta_T,\delta_X)$ is invariant under $\pi_{(\delta_T,\delta_X)}$, which ensures $\range\pi_{(\delta_T,\delta_X)} = \Span \Delta^*(\delta_T,\delta_X)$ and $\pi_{(\delta_T,\delta_X)}$ is a projection. Then
\begin{equation}\label{directsum}R(\delta_T,\delta_X)= \range \pi_{(\delta_T,\delta_X)} \oplus \ker \pi_{(\delta_T,\delta_X)}= \Span \Delta^*(\delta_T,\delta_X) \oplus \ker \pi_{(\delta_T,\delta_X)}. \end{equation}

\smallskip

By Proposition \ref{equi} and Lemma \ref{ajustement}, we have
\[ \forall \:  M\in \Mon(\delta_T,\delta_X), \:M - \pi_{(\delta_T,\delta_X)}(M) \in \mathcal{I}(\delta_T,\delta_X),\]
which proves the inclusion $\ker \pi_{(\delta_T,\delta_X)}=\range( \Id - \pi_{(\delta_T,\delta_X)}) \subset \mathcal{I}(\delta_T,\delta_X)$.

The proof is completed by Lemma \ref{zerocombi}, which implies that the family $\Delta^*(\delta_T,\delta_X)$ is linearly independent modulo $\mathcal{I}(\delta_T,\delta_X)$. It also implies
\[\mathcal{I}(\delta_T,\delta_X) \cap \Span(\Delta^*(\delta_T,\delta_X)=\{0\},\]
which entails the equality $\ker \pi_{(\delta_T,\delta_X)} = \mathcal{I}(\delta_T,\delta_X)$. Indeed, $\ker \pi_{(\delta_T,\delta_X)}$ is a complementary space of $\Span(\Delta^*(\delta_T,\delta_X))$ in $R(\delta_T,\delta_X)$ by (\ref{directsum}). Since $\ker \pi_{(\delta_T,\delta_X)}$ is included in  $\mathcal{I}(\delta_T,\delta_X)$, if the intersection of $\mathcal{I}(\delta_T,\delta_X)$ and $\Span(\Delta^*(\delta_T,\delta_X))$ is the nullspace then $\ker \pi_{(\delta_T,\delta_X)} = \mathcal{I}(\delta_T,\delta_X)$.
\end{proof}

Proposition \ref{span} displays $\Delta^*(\delta_T,\delta_X)$ as a set of representatives of $R(\delta_T,\delta_X)$ modulo $\mathcal{I}(\delta_T,\delta_X)$ and proves Theorem \ref{DIM}, which can be rephrased as follows.

\begin{corollary}\label{DIM2}
The dimension of the code $C_\eta(\delta_T,\delta_X)$ equals
\[\dim C_\eta(\delta_T,\delta_X) = \#\mathcal{K}^*(\delta_T,\delta_X).\]
\end{corollary}

\begin{proof}
It a straightforward consequence of Corollary \ref{nbDelta} and Theorem \ref{DIM}.
\end{proof}

\begin{example}\label{surj}
We can easily deduce from Corollary \ref{DIM2} that the evaluation map $\ev_{(\delta_T,\delta_X)}$ is surjective if $\delta_T \geq q$ and $\delta_X \geq q$. Indeed, in this case,
\[\mathcal{K}^*(\delta_T,\delta_X)=\mathcal{K}(\delta_T,\delta_X)=\left\lbrace  (\alpha,\beta) \in \N^2 \: \left| \:
\begin{array}{c} \alpha \in \lsem 0 , q-1 \rsem \cup \{\delta_X\} \\
\beta \in \lsem 0 , q-1 \rsem \cup \{\delta -\eta\alpha\}
\end{array}\right.\right\rbrace,\]
so that $\dim C_\eta(\delta_T,\delta_X) = \# \mathcal{K}^*(\delta_T,\delta_X)=(q+1)^2=N$.

\end{example}

\begin{figure}[h!]
\centering
\begin{tikzpicture}[scale=0.55]
    \clip (-2,-1) rectangle (5cm,11cm); 
    \draw[style=help lines,dashed,opacity=0.5] (-4,-17) grid[step=1cm] (14,17); 

 \draw [thick,-latex] (0,0) -- (0,11) node [above left] {$c_2$};
    \draw [thick,-latex] (0,0) -- (4,0) node [below right] {$d_2$};

    \draw [thick,blue] (0,0) -- (0,10 ) node [left,font=\small] {$\delta$};
    \draw [thick,blue] (0,0) -- (3,0) node [below, font=\small] {$\delta_X$} ;    
    \draw [thick,blue] (3,0) -- (3,4) node [below] {} ;
        \draw [thick,blue] (3,0) -- (3,4) node [below] {} ;
                \draw [thick,blue] (0,10) -- (3,4) node [below] {} ;
   \draw [thick,blue,dashed] (3,4) -- (0,4) node [left,font=\small] {$\delta_T$} ;

       \foreach \x in {0,1,2,3}{                           
      \foreach \y in {0,1,2,10-2*\x}{                       
        \node[draw,circle,inner sep=1pt,fill,black] at (\x,\y) {}; 
      }
    }  
    \foreach \y in {0,1}{                       
        \node[draw,circle,inner sep=1pt,fill,black] at (3,\y) {}; 
      } 
  \end{tikzpicture} 
\caption{Example of $\mathcal{P}(\delta_T,\delta_X)$ when $\ev_{(\delta_T,\delta_X)}$ is surjective: $\mathcal{P}(4,3)$ with $q=3$ in $\mathcal{H}_2$}\label{f3}
\end{figure}
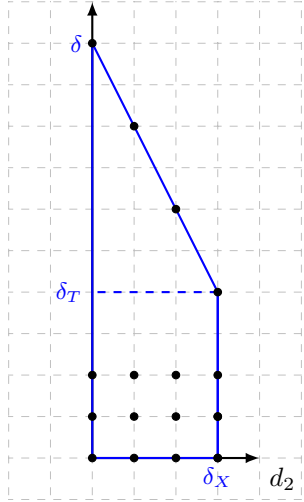

\subsection{Explicit formulae for the dimension of $C_\eta(\delta_T,\delta_X)$ and examples}\label{sectionS}

By Corollary \ref{DIM2}, computing the dimension is now reduced to the combinatorial question about the number of couples in $\mathcal{K}^*(\delta_T,\delta_X)$. The key of the proof of Proposition \ref{dimcode} below is to give a well-chosen partition of $\mathcal{K}^*(\delta_T,\delta_X)$ from which we can easily deduce its cardinality. Putting aside the very particular case of $\eta=0$, two cases have to be distinguished according to the sign of $\delta_T$, which determinates the shape of $\mathcal{P}(\delta_T,\delta_X)$ and the value of $A$. These two cases are themselves subdivided into several subcases, depending on the position of the preimage $s$ of $q$ under the function $ x \mapsto \delta - \eta x$ with respect to $\mathcal{A}_X$, defined in Notation \ref{notrep}.

\begin{prop}\label{dimcode}
On $\mathcal{H}_0$, the dimension of the evaluation code $C_0(\delta_T,\delta_X)$ equals
\[\dim C_0(\delta_T,\delta_X)= \left(\min(\delta_T,q) + 1 \right)\left(\min(\delta_X,q) + 1 \right).  \]

\smallskip

On $\mathcal{H}_\eta$ with $\eta \geq 2$, we set 
\[m = \min( \left\lfloor A\right\rfloor,q-1), \: h=\left\{\begin{array}{cl} \min(\delta_T,q)+1& \text{if } \delta_T \geq 0 \text{ and } q \leq \delta_X,\\
0 & \text{otherwise,} \end{array}\right.\]
\[s=\frac{\delta-q}{\eta} \text{ and }  \tilde{s}=\left\{\begin{array}{cl}
\lfloor s \rfloor &\text{if } s \in [0,m], \\
-1 & \text{if } s < 0,\\
m &\text{if } s > m.
\end{array}\right. \]
The evaluation code $C_\eta(\delta_T,\delta_X)$ on the Hirzebruch Surface $\mathcal{H}_\eta$ has dimension
\[ \dim C_\eta(\delta_T,\delta_X)=(q+1)(\tilde{s} +1) + (m-\tilde{s})\left(\delta+1-\eta \left(\frac{m + \tilde{s} +1}{2}\right)\right) + h.\]
\end{prop}

\begin{remark}
\begin{enumerate}
\item For $q$ large enough, the dimension is nothing but the number of points of $\mathcal{P}(\delta_T,\delta_X)$. This case was already studied by J. P. Hansen (\cite{HanToric} Theorem 1.5.)
\item The previous proposition generalizes the result of \cite{CN}, in which the authors studied rational scrolls. For $a_1 \geq a_2 \geq 1$, the rational scroll $S(a_1,a_2)$ is isomorphic to the Hirzebruch surface $\mathcal{H}_{(a_1-a_2)}$ (\cite{CoxToric} Example 3.1.16.). This geometric isomorphism induces a Hamming isometry between the codes. We thus can compare our result with theirs for $\eta=a_1-a_2$ and $\delta_T=a_2\delta_X$. Despite the appearing difference due to a different choice of monomial order (see Definition \ref{ordre} and Remark \ref{rkCN}), both formulae do coincide.
\end{enumerate}
\end{remark}%

\begin{proof}
\textbf{To prove the case $\eta=0$}, it is enough to write
\[\mathcal{K}^*(\delta_T,\delta_X)=\mathcal{K}(\delta_T,\delta_X)=\left\lbrace  (\alpha,\beta) \in \N^2 \: \left| \:
\begin{array}{c} \alpha \in \lsem 0 , \min(\delta_X,q)-1 \rsem \cup \{\delta_X\} \\
0 \leq \beta \leq \min(\delta_T,q)-1 \text{ or } \beta=\delta -\eta\alpha
\end{array}\right.\right\rbrace.\]

\smallskip

\textbf{Now, assume $\eta \geq 2$ and $\delta_T > 0$.} Notice that the sets $\mathcal{K}^*(\delta_T,\delta_X)$ and $\mathcal{K}(\delta_T,\delta_X)$ always coincide in this case.
\begin{itemize}
\item Let us assume that $q > \delta_X$.

\begin{itemize}
\item If $q > \delta$ also, then $s < 0$ and 
\[\mathcal{K}(\delta_T,\delta_X)= \bigcup_{\alpha=0}^{\delta_X} \{(\alpha,\beta) \: | \: \beta \in \lsem 0,\delta -\eta\alpha\rsem\}\]%
\begin{flalign*}\text{and thus }\#\mathcal{K}(\delta_T,\delta_X)&=\sum_{\alpha=0}^{\delta_X} (\delta -\eta\alpha+1)= (\delta_X+1)\left(\delta_T+\eta  \frac{\delta_X}{2}+1\right).& \end{flalign*}
\item If $ \delta_T \leq q \leq \delta$, then $0 \leq s \leq \delta_X$ and one can write
\[\begin{aligned}\mathcal{K}(\delta_T,\delta_X)=&\left(\bigcup_{\alpha=0}^{\lfloor s \rfloor} \{(\alpha,\beta) \: | \: \beta \in \lsem 0,q-1\rsem \cup\{\delta -\eta\alpha\}\}\right)\\
& \cup \left( \bigcup_{\alpha=\lfloor s \rfloor + 1}^{\delta_X} \{(\alpha,\beta) \: | \: \beta \in \lsem 0,\delta -\eta\alpha\rsem\} \right).\end{aligned}\]
\[ \begin{aligned}\text{and thus } \#\mathcal{K}(\delta_T,\delta_X)&=\sum_{\alpha=0}^{\lfloor s \rfloor} (q+1) + \sum_{\alpha=\lfloor s \rfloor + 1}^{\delta_X} (\delta +1 -\eta\alpha)&\\
&= (q+1)(\lfloor s \rfloor +1) + (\delta_X-\lfloor s \rfloor)\left(\delta+1-\eta \frac{\delta_X+ \lfloor s \rfloor+1}{2}\right).& \end{aligned}\]

\item If $\delta_X < q < \delta_T$, then $s > \delta_X$ and 
\[\mathcal{K}(\delta_T,\delta_X)=\left(\bigcup_{\alpha=0}^{\delta_X} \{(\alpha,\beta) \: | \: \beta \in \lsem 0,q-1\rsem \cup\{\delta -\eta\alpha\}\}\right)\]
and then $\#\mathcal{K}(\delta_T,\delta_X)=(q+1)(\delta_X+1)$.
\end{itemize}

\item Let us assume that $q \leq\delta_X$.

\begin{itemize}
\item If $\frac{\delta}{\eta +1} < q$, then $0 \leq s < q$ and $\lfloor s \rfloor \in \mathcal{A}_X$.
\[\begin{aligned}\mathcal{K}(\delta_T,\delta_X)=&\left(\bigcup_{\alpha=0}^{\lfloor s \rfloor} \{(\alpha,\beta) \: | \: \beta \in \lsem 0,q-1\rsem \cup\{\delta -\eta\alpha\}\}\right)\\
& \cup \left( \bigcup_{\alpha=\lfloor s \rfloor + 1}^{q-1} \{(\alpha,\beta) \: | \: \beta \in \lsem 0,\delta -\eta\alpha\rsem\} \right)\\
& \cup \{(\delta_X,\beta) \: | \: \beta \in \lsem 0 , h\rsem \}.\end{aligned}\]
Then
\[ \begin{aligned}\#\mathcal{K}(\delta_T,\delta_X)&=\sum_{\alpha=0}^{\lfloor s \rfloor} (q+1) + \sum_{\alpha=\lfloor s \rfloor + 1}^{q-1} (\delta+1 -\eta\alpha) + h+1.\\
&= (q+1)(\lfloor s \rfloor +1)+ (q-1-\lfloor s \rfloor)\left(\delta+1-\eta \frac{q + \lfloor s \rfloor}{2}\right) + h+1 \end{aligned}\]
\item If $q \leq \frac{\delta}{\eta +1}$, then $s \geq q$ and 
\[\mathcal{K}(\delta_T,\delta_X)=\left(\bigcup_{\alpha=0}^{q-1} \{(\alpha,\beta) \: | \: \beta \in \lsem 0,q-1\rsem \cup\{\delta -\eta\alpha\}\}\right) \cup \{(\delta_X,\beta) \: | \: \beta \in \lsem 0 , h\rsem \}.\]
Then $\#\mathcal{K}(\delta_T,\delta_X)=(q+1)q + h+1$.
\end{itemize}
\end{itemize}
\smallskip

\textbf{Finally assume $\eta \geq 2$ and $\delta_T \leq 0$.}

Let us rewrite $\mathcal{K}^*(\delta_T,\delta_X)$ to lead to formulae that coincide with the general one given above according to the position of $q$ in the increasing sequence 
\[\frac{\eta}{\eta+1} A < A \leq \eta A,\]
with $A=\frac{\delta}{\eta}$. For any $\alpha \in \mathcal{A}_X$,
\[q \leq \delta -\eta\alpha \Leftrightarrow \alpha \leq \frac{\delta - q }{\eta} = s < A.\]

\begin{itemize}
\item If $q > \eta A$, then $\mathcal{K}^*(\delta_T,\delta_X)=\mathcal{K}(\delta_T,\delta_X)$, $s < 0$ and we can write
\[\mathcal{K}(\delta_T,\delta_X)= \bigcup_{\alpha=0}^{\lfloor A \rfloor} \{(\alpha,\beta) \: | \: \beta \in \lsem 0,\delta -\eta\alpha\rsem\}\]%
\begin{flalign*}%
\text{and thus } \#\mathcal{K}(\delta_T,\delta_X)&=\sum_{\alpha=0}^{\lfloor A \rfloor} (\delta -\eta\alpha+1)= (\lfloor A \rfloor+1)\left(\delta+1 -\eta\frac{\lfloor A \rfloor}{2}\right).&
\end{flalign*}

\item If $A < q \leq \eta A$, we know that $\mathcal{K}^*(\delta_T,\delta_X)=\mathcal{K}(\delta_T,\delta_X)$ and we have
\[\mathcal{K}(\delta_T,\delta_X)=\left(\bigcup_{\alpha=0}^{\lfloor s \rfloor} \{(\alpha,\beta) \: | \: \beta \in \lsem 0,q-1\rsem \cup\{\delta -\eta\alpha\}\}\right) \cup \left( \bigcup_{\alpha=\lfloor s \rfloor + 1}^{\lfloor A \rfloor} \{(\alpha,\beta) \: | \: \beta \in \lsem 0,\delta -\eta\alpha\rsem\} \right)\]
and then
\[ \begin{aligned}\#\mathcal{K}(\delta_T,\delta_X)&=\sum_{\alpha=0}^{\lfloor s \rfloor} (q+1) + \sum_{\alpha=\lfloor s \rfloor + 1}^{\lfloor A \rfloor} (\delta -\eta\alpha+1)\\
&= (q+1)(\lfloor s \rfloor +1) + (\lfloor A \rfloor-\lfloor s \rfloor)\left(\delta+1-\eta \frac{\lfloor A \rfloor + \lfloor s \rfloor+1}{2}\right). \end{aligned}\]

\item If $\frac{\eta}{\eta+1}A < q \leq A$, then $q-1 < \lfloor A \rfloor$. Note that $\mathcal{K}^*(\delta_T,\delta_X) \neq \mathcal{K}(\delta_T,\delta_X)$ and

\[\begin{aligned}\mathcal{K}^*(\delta_T,\delta_X)=&\left(\bigcup_{\alpha=0}^{\lfloor s \rfloor} \{(\alpha,\beta) \: | \: \beta \in \lsem 0,q-1\rsem \cup\{\delta -\eta\alpha\}\}\right)\\
& \cup \left( \bigcup_{\alpha=\lfloor s \rfloor + 1}^{q-1} \{(\alpha,\beta) \: | \: \beta \in \lsem 0,\delta -\eta\alpha\rsem\} \right)\end{aligned}\]

Then
\[ \begin{aligned}\#\mathcal{K}^*(\delta_T,\delta_X)&=\sum_{\alpha=0}^{\lfloor s \rfloor} (q+1) + \sum_{\alpha=\lfloor s \rfloor + 1}^{q-1} (\delta -\eta\alpha+1)\\
&= (q+1)(\lfloor s \rfloor +1) + (q-1-\lfloor s \rfloor)\left(\delta+1-\eta \frac{q + \lfloor s \rfloor}{2}\right) \end{aligned}\]

\item If $q \leq \frac{\eta}{\eta +1}A$, then $s \geq q$, $\mathcal{K}^*(\delta_T,\delta_X) \neq \mathcal{K}(\delta_T,\delta_X)$ and 
\[\mathcal{K}^*(\delta_T,\delta_X)=\left(\bigcup_{\alpha=0}^{q-1} \{(\alpha,\beta) \: | \: \beta \in \lsem 0,q-1\rsem \cup\{\delta -\eta\alpha\}\}\right)\]
which gives $\#\mathcal{K}^*(\delta_T,\delta_X)=(q+1)q$.
\end{itemize}
\end{proof}

\subsection{Examples}

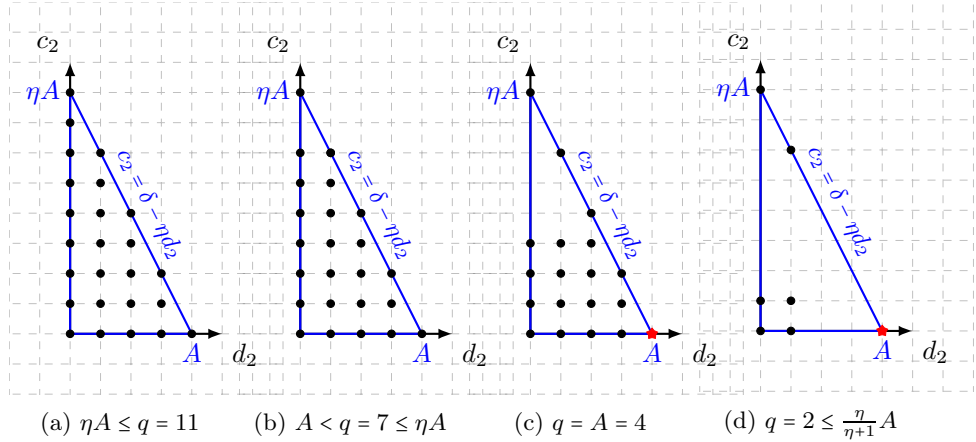
\begin{figure}[h!]

\begin{subfigure}[b]{0.24\textwidth}
\centering
  \begin{tikzpicture}[scale=0.4]

    \clip (-2,-2) rectangle (7cm,11cm); 

    \draw[style=help lines,dashed,opacity=0.5] (-4,-10) grid[step=1cm] (14,15); 

 \draw [thick,-latex] (0,0) -- (0,9) node [above left] {$c_2$};
    \draw [thick,-latex] (0,0) -- (5,0) node [below right] {$d_2$};
        
    \draw [thick,blue] (0,0) -- (0,8) node [left] {$\eta A$};
    \draw [thick,blue] (0,0) -- (4,0) node [below] {$A$} ;
   \draw [thick,blue] (0,8) -- (4,0) node [midway,above,sloped,font=\small] {$c_2=\delta-\eta d_2$} ;

   \node[draw,circle,inner sep=1pt,fill,black] at (4,0) {};
    \foreach \y in {0,1,2,3,4,5,6,7,8}{                       
        \node[draw,circle,inner sep=1pt,fill,black] at (0,\y) {}; 
        }
    \foreach \y in {0,1,2,3,4,5,6}{                       
        \node[draw,circle,inner sep=1pt,fill,black] at (1,\y) {}; 
        }
    \foreach \y in {0,1,2,3,4}{                       
        \node[draw,circle,inner sep=1pt,fill,black] at (2,\y) {}; 
        }
    \foreach \y in {0,1,2}{                       
        \node[draw,circle,inner sep=1pt,fill,black] at (3,\y) {}; 
        }
  \end{tikzpicture}
  \subcaption{$\eta A \leq q=11$}
\end{subfigure} 
\begin{subfigure}[b]{0.24\textwidth}
\centering
  \begin{tikzpicture}[scale=0.4]

    \clip (-2,-2) rectangle (7cm,11cm); 

    \draw[style=help lines,dashed,opacity=0.5] (-4,-10) grid[step=1cm] (14,15); 

 \draw [thick,-latex] (0,0) -- (0,9) node [above left] {$c_2$};
    \draw [thick,-latex] (0,0) -- (5,0) node [below right] {$d_2$};
        
    \draw [thick,blue] (0,0) -- (0,8) node [left] {$\eta A$};
    \draw [thick,blue] (0,0) -- (4,0) node [below] {$A$} ;
   \draw [thick,blue] (0,8) -- (4,0) node [midway,above,sloped,font=\small] {$c_2=\delta-\eta d_2$} ;
   
       \foreach \x in {0,1}{                           
      \foreach \y in {0,1,2,3,4,5,6,8-2*\x}{                       
        \node[draw,circle,inner sep=1pt,fill,black] at (\x,\y) {}; 
      }
    }   
   \node[draw,circle,inner sep=1pt,fill,black] at (4,0) {};
    \foreach \y in {0,1,2,3,4}{                       
        \node[draw,circle,inner sep=1pt,fill,black] at (2,\y) {}; 
        }
    \foreach \y in {0,1,2}{                       
        \node[draw,circle,inner sep=1pt,fill,black] at (3,\y) {}; 
        }
  \end{tikzpicture}
  \subcaption{$A < q = 7 \leq \eta A$}
\end{subfigure} 
\begin{subfigure}[b]{0.24\textwidth}
\centering
  \begin{tikzpicture}[scale=0.4]

    \clip (-2,-2) rectangle (7cm,11cm); 

    \draw[style=help lines,dashed,opacity=0.5] (-4,-10) grid[step=1cm] (14,15); 

 \draw [thick,-latex] (0,0) -- (0,9) node [above left] {$c_2$};
    \draw [thick,-latex] (0,0) -- (5,0) node [below right] {$d_2$};
        
    \draw [thick,blue] (0,0) -- (0,8) node [left] {$\eta A$};
    \draw [thick,blue] (0,0) -- (4,0) node [below] {$A$} ;
   \draw [thick,blue] (0,8) -- (4,0) node [midway,above,sloped,font=\small] {$c_2=\delta-\eta d_2$} ;
  
       \foreach \x in {0,1,2}{                           
      \foreach \y in {0,1,2,3,8-2*\x}{                       
        \node[draw,circle,inner sep=1pt,fill,black] at (\x,\y) {}; 
      }
    }   
   \node[draw,star,inner sep=1pt,fill,red] at (4,0) {};
    \foreach \y in {0,1,2}{                       
        \node[draw,circle,inner sep=1pt,fill,black] at (3,\y) {}; 
        }
  \end{tikzpicture}
  \subcaption{$q = A=4$}
\end{subfigure} 
\begin{subfigure}[b]{0.24\textwidth}
\centering
  \begin{tikzpicture}[scale=0.4]

    \clip (-2,-2) rectangle (7cm,11cm); 

    \draw[style=help lines,dashed,opacity=0.5] (-4,-10) grid[step=1cm] (14,15); 

 \draw [thick,-latex] (0,0) -- (0,9) node [above left] {$c_2$};
    \draw [thick,-latex] (0,0) -- (5,0) node [below right] {$d_2$};
        
    \draw [thick,blue] (0,0) -- (0,8) node [left] {$\eta A$};
    \draw [thick,blue] (0,0) -- (4,0) node [below] {$A$} ;
   \draw [thick,blue] (0,8) -- (4,0) node [midway,above,sloped,font=\small] {$c_2=\delta-\eta d_2$} ;
  
       \foreach \x in {0,1}{                           
      \foreach \y in {0,1,8-2*\x}{                       
        \node[draw,circle,inner sep=1pt,fill,black] at (\x,\y) {}; 
      }
    }   
   \node[draw,star,inner sep=1pt, fill,red] at (4,0) {};
  \end{tikzpicture}
  \subcaption{$q =2 \leq \frac{\eta}{\eta+1} A$}
\end{subfigure}  
 
\caption{$\mathcal{P}(-2,5)$ in $\mathcal{H}_2$ for different values for $q$.}\label{f5}
\end{figure}

\begin{example}
Let us compute the dimension of the code $C_2(-2,5)$ using the previous formula on different finite fields. We have $A=4 \in \N$. Beware that $\eta$ divides $\delta _T$, so (\ref{H}) may hold. See Figure \ref{f5}.
\begin{itemize}
\item On $\F_{11}$, $m=A, \: s<0, \: \tilde{s}=-1$,
\[\dim C_2(-2,5)=(4+1)\left(-2+2\left(5-\frac{4}{2}\right)+1\right)=25. \]
\item On $\F_7$, $m=A, \: s=\tilde{s}=0$, 
\[ \dim C_2(-2,5)=(7+1)+4\left(-2+2\left(5-\frac{5}{2}\right)+1\right)=8+16=24. \]
\item On $\F_4$, $m=3, \: s=\tilde{s}=2$. Then (\ref{H}) holds and 
\[\dim C_2(-2,5)=(4+1)(2+1)+\left(-2+2\left(5-\frac{6}{2}\right)+1\right)=15+3=18. \]
\item On $\F_2$, $m=1, \: s=\tilde{s}=1$. Then (\ref{H}) holds and
\[ \dim C_2(-2,5)=(2+1)(1+1)=6. \]
\end{itemize}
\end{example}

\begin{example}
To illustrate the cases $q \leq \delta_X$, let us compute the dimension of the code $C_2(1,3)$ using the previous formula on $\F_3$ and $\F_2$. See Figure \ref{f6}.
\begin{itemize}
\item On $\F_3$, $m=2, \: s=\tilde{s}=2, \: h=1$, $\dim C_2(1,3)=(3+1)(2+1)+1+1=14$.
\item On $\F_2$, $m=1, \: s=2.5>m,  \:\tilde{s}=1, \: h=1$, $\dim C_2(-2,5)=(2+1)(1+1)+1+1=6+1=8$.
\end{itemize}
\end{example}

\begin{figure}[h!]
\begin{subfigure}[b]{0.48\textwidth}
\centering
\begin{tikzpicture}[scale=0.6]
    \clip (-3,-1) rectangle (6cm,9cm); 
    \draw[style=help lines,dashed,opacity=0.5] (-4,-17) grid[step=1cm] (14,17); 

 \draw [thick,-latex] (0,0) -- (0,8) node [above left] {$c_2$};
    \draw [thick,-latex] (0,0) -- (4,0) node [below right] {$d_2$};

    \draw [thick,blue] (0,0) -- (0,7 ) node [left,font=\small] {$\delta$};
    \draw [thick,blue] (0,0) -- (3,0) node [below, font=\small] {$\delta_X$} ;    
    \draw [thick,blue] (3,0) -- (3,1) node [below] {} ;
        \draw [thick,blue] (3,0) -- (3,1) node [below] {} ;
                \draw [thick,blue] (0,7) -- (3,1) node [midway,above,sloped] {$c_2=\delta- \eta d_2$} ;
   \draw [thick,blue,dashed] (3,1) -- (0,1) node [left,font=\small] {$\delta_T$} ;

       \foreach \x in {0,1,2}{                           
      \foreach \y in {0,1,2,7-2*\x}{                       
        \node[draw,circle,inner sep=1pt,fill,black] at (\x,\y) {}; 
      }
    }  
    \foreach \y in {0,1}{                       
        \node[draw,circle,inner sep=1pt,fill,black] at (3,\y) {}; 
      } 
  \end{tikzpicture} 
\subcaption{$\frac{\delta}{\eta+1} < q=\delta_X=3$}
\end{subfigure}%
\begin{subfigure}[b]{0.48\textwidth}
\centering
\begin{tikzpicture}[scale=0.6]
    \clip (-3,-1) rectangle (6cm,9cm); 
    \draw[style=help lines,dashed,opacity=0.5] (-4,-17) grid[step=1cm] (14,17); 

 \draw [thick,-latex] (0,0) -- (0,8) node [above left] {$c_2$};
    \draw [thick,-latex] (0,0) -- (4,0) node [below right] {$d_2$};

    \draw [thick,blue] (0,0) -- (0,7 ) node [left,font=\small] {$\delta$};
    \draw [thick,blue] (0,0) -- (3,0) node [below, font=\small] {$\delta_X$} ;    
    \draw [thick,blue] (3,0) -- (3,1) node [below] {} ;
        \draw [thick,blue] (3,0) -- (3,1) node [below] {} ;
                \draw [thick,blue] (0,7) -- (3,1) node [midway,above,sloped] {$c_2=\delta- \eta d_2$} ;
   \draw [thick,blue,dashed] (3,1) -- (0,1) node [left,font=\small] {$\delta_T$} ;

       \foreach \x in {0,1,3}{                           
      \foreach \y in {0,1,7-2*\x}{                       
        \node[draw,circle,inner sep=1pt,fill,black] at (\x,\y) {}; 
      }
    }   
  \end{tikzpicture} 
\subcaption{$q=2 \leq \frac{\delta}{\eta+1}$}
\end{subfigure}
\caption{$\mathcal{P}(1,3)$ in $\mathcal{H}_2$}\label{f6}
\end{figure}

\begin{figure}[h!]
\begin{subfigure}[b]{0.32\textwidth}
\begin{tikzpicture}[scale=0.5]
    \clip (-3,-1) rectangle (5cm,13cm); 
    \draw[style=help lines,dashed,opacity=0.5] (-4,-17) grid[step=1cm] (14,17); 

 \draw [thick,-latex] (0,0) -- (0,12) node [above left] {$c_2$};
    \draw [thick,-latex] (0,0) -- (4,0) node [below right] {$d_2$};

    \draw [thick,blue] (0,0) -- (0,11 ) node [left,font=\small] {$\delta$};
    \draw [thick,blue] (0,0) -- (3,0) node [below,font=\small] {$\delta_X$} ;    
    \draw [thick,blue] (3,0) -- (3,5) node [below] {} ;
       \draw [thick,blue,dashed] (3,5) -- (0,5) node [left,font=\small] {$\delta_T$} ;
   \draw [thick,blue] (0,11) -- (3,5) node [midway,above,sloped,font=\small] {$c_2=\delta- \eta d_2$} ;
          \foreach \y in {0,1,2,3,4,5,6,7,8,9,10,11}{                       
        \node[draw,circle,inner sep=1pt,fill,black] at (0,\y) {}; 
      } 
          \foreach \y in {0,1,2,3,4,5,6,7,8,9}{                       
        \node[draw,circle,inner sep=1pt,fill,black] at (1,\y) {}; 
      } 
          \foreach \y in {0,1,2,3,4,5,6,7}{                       
        \node[draw,circle,inner sep=1pt,fill,black] at (2,\y) {}; 
      } 
          \foreach \y in {0,1,2,3,4,5}{                       
        \node[draw,circle,inner sep=1pt,fill,black] at (3,\y) {}; 
      } 
  \end{tikzpicture} 
\subcaption{$\delta < q=13$}
\end{subfigure}
\begin{subfigure}[b]{0.32\textwidth}
\begin{tikzpicture}[scale=0.5]
    \clip (-3,-1) rectangle (5cm,13cm); 
    \draw[style=help lines,dashed,opacity=0.5] (-4,-17) grid[step=1cm] (14,17); 

 \draw [thick,-latex] (0,0) -- (0,12) node [above left] {$c_2$};
    \draw [thick,-latex] (0,0) -- (4,0) node [below right] {$d_2$};

    \draw [thick,blue] (0,0) -- (0,11 ) node [left,font=\small] {$\delta$};
    \draw [thick,blue] (0,0) -- (3,0) node [below,font=\small] {$\delta_X$} ;    
    \draw [thick,blue] (3,0) -- (3,5) node [below] {} ;
       \draw [thick,blue,dashed] (3,5) -- (0,5) node [left,font=\small] {$\delta_T$} ;
   \draw [thick,blue] (0,11) -- (3,5) node [midway,above,sloped,font=\small] {$c_2=\delta- \eta d_2$} ;
       \foreach \x in {0,1,2}{                           
      \foreach \y in {0,1,2,3,4,5,6,11-2*\x}{                       
        \node[draw,circle,inner sep=1pt,fill,black] at (\x,\y) {}; 
      }
    }  
          \foreach \y in {0,1,2,3,4,5}{                       
        \node[draw,circle,inner sep=1pt,fill,black] at (3,\y) {}; 
      } 
  \end{tikzpicture} 
\subcaption{$\delta_T \leq q=7 \leq \delta $}
\end{subfigure}
\begin{subfigure}[b]{0.32\textwidth}
\begin{tikzpicture}[scale=0.5]
    \clip (-3,-1) rectangle (5cm,13cm); 
    \draw[style=help lines,dashed,opacity=0.5] (-4,-17) grid[step=1cm] (14,17); 

 \draw [thick,-latex] (0,0) -- (0,12) node [above left] {$c_2$};
    \draw [thick,-latex] (0,0) -- (4,0) node [below right] {$d_2$};

    \draw [thick,blue] (0,0) -- (0,11 ) node [left,font=\small] {$\delta$};
    \draw [thick,blue] (0,0) -- (3,0) node [below,font=\small] {$\delta_X$} ;    
    \draw [thick,blue] (3,0) -- (3,5) node [below] {} ;
       \draw [thick,blue,dashed] (3,5) -- (0,5) node [left,font=\small] {$\delta_T$} ;
   \draw [thick,blue] (0,11) -- (3,5) node [midway,above,sloped,font=\small] {$c_2=\delta- \eta d_2$} ;
       \foreach \x in {0,1,2,3}{                           
      \foreach \y in {0,1,2,3,11-2*\x}{                       
        \node[draw,circle,inner sep=1pt,fill,black] at (\x,\y) {}; 
      }
    }   
  \end{tikzpicture} 
\subcaption{$\delta_X<q=4< \delta_T$}
\end{subfigure}
\caption{$\mathcal{P}(5,3)$ in $\mathcal{H}_2$}\label{f7}
\end{figure}

\begin{example}
On $\mathcal{H}_2$, let us compute the dimension of the code $C_2(5,3)$ on the finite fields $\F_{13}$, $\F_7$ and $\F_4$. See Figure \ref{f7}. Since $q > \delta_X$, we have $m=\delta_X=3$. 
\begin{itemize}
\item On $\F_{13}$, $s<0$ then $\tilde{s}=-1$.
\[\dim C_2(5,3)=(3+1)(5+3+1)=36.\]
\item On $\F_7$, $s=\tilde{s}=2$.
\[\dim C_2(5,3)=(7+1)(2+1)+(3-2)\left(5+2\left(3-\frac{3+2+1}{2}\right)+1\right)=24+6=30.\]
\item On $\F_4$, $s>m$ then $\tilde{s}=m=3$.
\[\dim C_2(5,3)=(4+1)(3+1)=20.\]
\end{itemize}
\end{example}

\section{Gr\"obner Basis}\label{secgrob}

Our strategy to compute the dimension of the code highlights the key role of monomials in our study. Monomials remain crucial in the calculus of the minimum distance, through the use of Gr\"obner bases. Until now, every technique we used has come from linear algebra, focusing on the finite dimensional vector spaces $R(\delta_T,\delta_X)$ and vector subspaces $\ker \ev_{(\delta_T,\delta_X)}$. However considering a convenient ideal of the ring $R$ gives the possibility of using algebraic tools, Gr\"obner bases theory here, to handle the minimum distance problem.

\smallskip

Let us first recall classical facts about Gr\"obner bases. The reference for this section is \cite{CLO}.

Let $R$ be a polynomial ring. A \emph{monomial order} is a total order on the monomials, denoted by $<$, satisfying the following compatibility property with multiplication: for all monomials $M, \: N, \: P$,
\[ M < N \: \Rightarrow \: MP < NP \text{ and } M < MP.\]
For every polynomial $F \in R$, one can define the \emph{leading monomial} of $F$, denoted by $\LM(F)$, to be the greatest monomial for this ordering that appears in $F$. The \emph{leading term} of $F$ is denoted by $\LT(F)$ and is defined as the leading monomial of $F$ multiplied by its coefficient in $F$.

Let $I$ be an ideal of the polynomial ring $R$, endowed with a monomial order $<$. The \emph{monomial ideal} $\LT(I) \subset R$ associated to $I$ is the ideal generated by the leading terms $\LT(F)$ of all polynomials $F \in I$.  A finite subset $G$ of an ideal $I$ is a \emph{Gr\"obner basis} of the ideal $I$ if $\LT(I)=\left\langle \LT(g) \: | \:  g \in G \right\rangle$.

The pleasing property of Gr\"obner bases (see \cite{Sturm} Proposition 1.1) that will be used to compute the minimum distance of the code is the following.

\begin{prop}\label{Grob}
Let $I$ be an ideal of a polynomial ring $R$ with Gr\"obner basis $G$. Then, setting $\pi$ as the canonical projection of $R$ onto $\faktor{R}{I}$, the set
\[\{\pi(M) \: | \: M \text{ monomials of } R \text{ such that for all } g \in G, \: \LT(g) \nmid M\}\]
is a basis of $\faktor{R}{I}$ as a vector space.
\end{prop}

\smallskip

Now that the necessary background is set up, let us define the ideal we shall use here.

\begin{notation}
Set $\displaystyle \mathcal{I}=\bigoplus_{(\delta_T,\delta_X) \in \Z \times \N} \mathcal{I}(\delta_T,\delta_X)$.
\end{notation}

Therefore, the ideal $\mathcal{I}$ is homogeneous : whenever it contains an element, it also contains all the homogeneous components of this element. This entails that  $\mathcal{I}$ is the homogeneous vanishing ideal of the subvariety consisting of the $\F_q$-rational points of the Hirzebruch surface $\mathcal{H}_\eta$.

\smallskip

Another ingredient to benefit from Gr\"obner bases theory is a suitable monomial order over $R=\F_q[T_1,T_2,X_1,X_2]$.

\begin{definition}\label{ordre}
Let us define a order on monomials of $R$ by stating that
\[T_1^{c'_1}T_2^{c'_2}X_1^{d'_1}X_2^{d'_2} < T_1^{c_1}T_2^{c_2}X_1^{d_1}X_2^{d_2}\]
if and only if
\[d'_1+d'_2 < d_1+d_2 \text{ or } \left\vert\begin{array}{rl}d'_1+d'_2&=d_1+d_2 \\ d'_2 &< d_2\end{array}\right. \text{ or } \left\vert\begin{array}{cl}d'_1&=d_1 \\ d'_2 &= d_2\\
c'_2&<c_2\end{array}\right. \text{ or } \left\vert\begin{array}{cl}d'_1&=d_1 \\ d'_2 &= d_2\\
c'_2&=c_2\\c'_1& < c_1\end{array}\right.\]
\end{definition}
One can easily check that $<$ is a monomial order.

\begin{remark}\label{rkCN}
Notice that exchanging the role of $d_1$ and $d_2$ and the one of $c_1$ and $c_2$, we recover the monomial order chosen by Carvalho and Neumann \cite{CN}.
\end{remark}

The choice of this monomial order is motivated by the choice of representatives of monomials under $\equiv$, hence by the choice of the  projection map $\pi_{(\delta_T,\delta_X)}$, as stated by the following lemma.
\begin{lemma}\label{lt}
Any monomial $M\in \Mon(\delta_T,\delta_X)$ is greater than the leading term of its image under $\pi_{(\delta_T,\delta_X)}$.
\end{lemma}

\begin{proof}
Since any $M \in \Mon(\delta_T,\delta_X)$ and its image under $\pi_{(\delta_T,\delta_X)}$ have the same bidegree, the first case of Definition \ref{ordre} never occurs.

Except if (\ref{H}) holds and $(d_2,c_2)=\left(\frac{\delta}{\eta},0\right)$, the image of a monomial $M(d_2,c_2) \in \Mon(\delta_T,\delta_X)$ under $\pi_{(\delta_T,\delta_X)}$ is the monomial $M(p_{(\delta_T,\delta_X)}(d_2,c_2))$, where $p_{(\delta_T,\delta_X)}$ is given in Definition \ref{proj}.

Write $(d'_2,c'_2)=p_{(\delta_T,\delta_X)}(d_2,c_2)$. Then $d_2 \geq d_2'$. If $d_2=d'_2$, then $c_2 \geq c'_2$, which means that $M \geq \pi_{(\delta_T,\delta_X)}(M)$.

\smallskip

It remains to check that it is also true for $M=M\left(\frac{\delta}{\eta},0\right)$ when (\ref{H}) holds. In this case, according to Definition \ref{pi}, 
\[\pi_{(\delta_T,\delta_X)}\left(M\left(\frac{\delta}{\eta},0\right)\right)= M(r,0)+M(r,\eta k(q-1))-M(r,q-1).\]
with $r < \frac{\delta}{\eta}$. Then $\LT\left(\pi_{(\delta_T,\delta_X)}\left(M\left(\frac{\delta}{\eta},0\right)\right)\right) < M\left(\frac{\delta}{\eta},0\right)$, which concludes the proof.
\end{proof}

\begin{notation}
Let $(\delta_T,\delta_X) \in \Z \times \N$. Let us set 
\[\mathcal{G}(\delta_T,\delta_X)=\{M-\pi_{(\delta_T,\delta_X)}(M) \: | \: M \in \Mon(\delta_T,\delta_X)\setminus  \Delta^\star(\delta_T,\delta_X)\}\]
and
\[\mathcal{G}=\bigcup_{(\delta_T,\delta_X)} \mathcal{G}(\delta_T,\delta_X).\]
\end{notation}

\begin{prop}\label{grob}
There exists a finite subset $\mathcal{G}'$ of $\mathcal{G}$ that forms a Gr\"obner basis of the ideal $\mathcal{I}$.
\end{prop}

\begin{proof}
First, let us prove that the leading term of any polynomial of $\mathcal{I}$ is divisible by the leading term of an element of $\mathcal{G}$.

Fix $f \in \mathcal{I}$. We write $f_{(\delta_T,\delta_X)}$ the homogeneous component of $f$ that has bidegree $(\delta_T,\delta_X)$.
The leading term of $f$ is the leading term of one of its homogeneous component. Therefore, it is enough to prove that the leading term of any $f_{(\delta_T,\delta_X)}$  is divisible by the leading term of an element of $\mathcal{G}$.

By Proposition \ref{span}, the map $\pi_{(\delta_T,\delta_X)}$ is a projection along $\mathcal{I}(\delta_T,\delta_X)$ onto $\Span \Delta^*(\delta_T,\delta_X)$. Hence $\ker \pi_{(\delta_T,\delta_X)}=\mathcal{I}(\delta_T,\delta_X)=\range(\Id - \pi_{(\delta_T,\delta_X)})$. The set $\mathcal{G}(\delta_T,\delta_X)$ is thus a spanning family for the vector space $\mathcal{I}(\delta_T,\delta_X)$. Therefore any $f_{(\delta_T,\delta_X)}$ can be written as a linear combination of elements of $\mathcal{G}(\delta_T,\delta_X)$:
\[f_{(\delta_T,\delta_X)} = \sum_{\substack{M \in \Mon(\delta_T,\delta_X)\\ M \notin \Delta^\star(\delta_T,\delta_X)}} c_M \left(M-\pi_{(\delta_T,\delta_X)}(M)\right)\]
By Lemma \ref{lt}, the leading monomial of $f_{(\delta_T,\delta_X)}$ is the maximum monomial $M_{max}$ with respect to the monomial order $<$ among the monomials $M$ in $\Mon(\delta_T,\delta_X)\setminus  \Delta^\star(\delta_T,\delta_X)$ such that $c_M \neq 0$. It is thus clear that the leading term of $f$ is divisible by the leading term of $M_{max}- \pi_{(\delta_T,\delta_X)}\left(M_{max}\right)$, that belongs to $\mathcal{G}$. 

\smallskip

To conclude, it is enough to apply Dickson's Lemma (\cite{CLO} \S 4 Theorem 4) to the monomial ideal $\LT(\mathcal{I})$.
\end{proof}

Let us highlight that the homogeneity of the ideal $\mathcal{I}$ gives a natural graduation of the quotient 
\[\faktor{R}{I}=\bigoplus_{(\delta_T,\delta_X)} \left(\faktor{R(\delta_T,\delta_X)}{\mathcal{I}(\delta_T,\delta_X)}\right),\]
from which, with Propositions \ref{Grob} and \ref{grob}, the next corollary arises. 

\begin{corollary}\label{base}
Let $(\delta_T,\delta_X) \in \Z \times \N$ such that $\delta \geq 0$. A basis of a complementary space of $\mathcal{I}(\delta_T,\delta_X)$ in $R(\delta_T,\delta_X)$ is the set $\{ M \in \Mon(\delta_T,\delta_X) \: | \: \forall \: g \in \mathcal{G}', \: \LT(g) \nmid M\}$.
\end{corollary}
\begin{proof}
By Propositions \ref{grob} and \ref{Grob}, the set 
\[\mathcal{B}=\left\{ M \in \bigcup_{(\delta_T,\delta_X) \in \Z\times \N} \Mon(\delta_T,\delta_X) \: | \: \forall \: g \in \mathcal{G}', \: \LT(g) \nmid M\right\}\]
is a basis of a complementary of $\mathcal{I}$, seen as $\F_q$-vector subspace of $R$. Every element of $\mathcal{B}$ is homogeneous. The result consists only on restricting on a homogeneous component.
\end{proof}

In Proposition \ref{span} we displayed $\Delta^*(\delta_T,\delta_X)$ as a basis of $R(\delta_T,\delta_X)$ modulo the subspace $\mathcal{I}(\delta_T,\delta_X)$ for each couple $(\delta_T,\delta_X)$. Actually the image under the canonical projection of the union of the $\Delta^*(\delta_T,\delta_X)$ is exactly the basis given by the previous proposition, as stated in the following lemma.

\begin{lemma}\label{repDiv}
Let us set 
$\Delta^*=\bigcup_{(\delta_T,\delta_X)} \Delta^*(\delta_T,\delta_X)$.
Then $\Delta^*$ is the set of monomials of $R$ that are not divisible by the leading term of any polynomial of $\mathcal{G}$.
\end{lemma}

\begin{proof} Fix $(\alpha,\beta) \in\mathcal{K}^*(\delta_T,\delta_X)$ and $M=M(\alpha,\beta) \in \Delta^*(\delta_T,\delta_X)$.

\smallskip

Let $G=N - \pi_{(\epsilon_T,\epsilon_X)}(N) \in \mathcal{G}(\epsilon_T,\epsilon_X)$ with 
\[N=T_1^{c_1}T_2^{c_2}X_1^{d_1}X_2^{d_2} \in \Mon(\epsilon_T,\epsilon_X) \setminus \Delta^*(\epsilon_T,\epsilon_X).\]

By Lemma \ref{lt}, $\LT(G)=N$.

Assume that $N$ divides $M$, that is to say
\begin{align}
\tag{i} \label{i} d_2 &\leq \alpha \\
\tag{ii} \label{ii} d_1=\epsilon_X-d_2 &\leq \delta_X-\alpha\\
\tag{iii} \label{iii} c_2 &\leq \beta \\
\tag{iv} \label{iv} c_1=\epsilon_T+\eta(\epsilon_X-d_2)-c_2 & \leq \delta -\eta\alpha-\beta
\end{align}
We want to reach a contradiction.

\smallskip

First suppose that $\pi_{(\epsilon_T,\epsilon_X)}(N)= T_1^{c'_1}T_2^{c'_2}X_1^{d'_1}X_2^{d'_2}$ is a monomial. By Lemma \ref{traduc}, since $N \equiv \pi_{(\epsilon_T,\epsilon_X)}(N)$, there exist $k, \: l \in \Z$ such that
\[d_2=d'_2+k(q-1), \: d_1=d'_1-k(q-1), \: c_2=c'_2+l(q-1)\text{ and } c_1 = c'_1 - (l + \eta k)(q-1).\]
Since $\LT(G)=N$, either $d'_2 > d_2$ or $d'_2=d_2$ and $c'_2 > c_2$, i.e either $k \in \N^*$ or $k=0$ and $l \in \N^*$. 

\begin{itemize}
\item Let first assume that $k \in \N^*$. By condition (\ref{D1}) for $i=2$, this implies that $d'_2 \geq 1$ and then $d_2 \geq q$. By (\ref{i}), the only possible value for $\alpha$ is thus $\alpha = A$ if $A \geq q$.
\begin{itemize}
\item If $\delta_T \geq 0$, then $\alpha=A=\delta_X$ and then, by (\ref{ii}), $d_1=0$. By condition (\ref{D3}) for $i=1$, $d'_1=0$, which implies $k=0$ and leads to a contradiction.
\item If $\delta_T < 0$, there is no integer $\alpha \geq q$ such that there exists $\beta \in \N$ satisfying $(\alpha,\beta) \in \mathcal{K}^*(\delta_T,\delta_X)$. This case never occurs.
\end{itemize}
\item Now, let us assume that $k=0$ and $l \in \N^*$, which implies $c_2 \geq q$. Since  $c_2 \leq \beta$, by Notation \ref{notrep}, $\beta=\delta_T+\eta (\delta_X-\alpha)$. Then $c_1 =0$ hence $c'_1=0$, which contradicts the hypothesis $l \neq 0$.
\end{itemize}

\smallskip

Now assume that $(\epsilon_T,\epsilon_X)$ satisfies (\ref{H}) and $N=X_1^{d_1}X_2^{d_2}$ with $d_1=-\frac{\epsilon_T}{\eta}\neq 0$ and $d_2=\epsilon_X+\frac{\epsilon_T}{\eta} \geq q$.
As before, by (\ref{i}), $\alpha$ can only be equal to $A$ if $A\geq q$, which happens only if $\delta_T \geq 0$ and $A=\delta_X$. The same reasoning than previously leads to a contradiction.

We then have proved that $\Delta^*$ is a subset of the set of monomials non divisible by the leading term of any polynomial in $\mathcal{G}$. But these two sets are basis of two complementary spaces of a same vector space by Proposition \ref{span} and Corollary \ref{base}. Therefore, these two sets coincide.
\end{proof}

\section{Minimum distance of $C_\eta(\delta_T,\delta_X)$}\label{secdist}

\subsection{Proof of the the lower bound of the minimum distance in Theorem \ref{DIST}}

Let us fix $(\epsilon_T,\epsilon_X) \in \N^2$ such that $\epsilon_T, \: \epsilon_X \geq q$.

\begin{notation}\label{DeltaM}
Let us set
\[\Delta^*(\epsilon_T,\epsilon_X)_F=\left\{ N \in \Delta^*(\epsilon_T,\epsilon_X) \: | \: \LT(F) \: | \: N\right\}\]
with $\Delta^*(\epsilon_T,\epsilon_X)$ defined in Notations \ref{notrep} and \ref{etoile}.
\end{notation}

Let $F \in R(\delta_T,\delta_X) \setminus \ker \ev_{(\delta_T,\delta_X)}$ and $\mathcal{Z}(F)$ its zero set in $\mathcal{H}_\eta$. We define
\[N_F=\# \mathcal{Z}(F)(\F_q).\]

We prove now the lower bound
\[d_\eta(\delta_T,\delta_X)  \geq \min_{M \in \Delta^*(\delta_T,\delta_X)} \# \Delta^*(\epsilon_T,\epsilon_X)_M.\]


\begin{proof}[Proof of the lower bound]
Recall that the minimum distance is defined by
\[d_\eta(\delta_T,\delta_X)=\min_{\substack{F \in R(\delta_T,\delta_X)\\ F \notin \ker \ev_{(\delta_T,\delta_X)}}} \omega(\ev_{(\delta_T,\delta_X)}(F)).\]

\medskip

First, Proposition \ref{span} gives $\ev_{(\delta_T,\delta_X)}(F)=\ev_{(\delta_T,\delta_X)}\left( \pi_{(\delta_T,\delta_X)}(F)\right)$ and then
\[d_\eta(\delta_T,\delta_X)=\min_{F \in \Span \Delta^*(\delta_T,\delta_X)} \omega(\ev_{(\delta_T,\delta_X)}(F))=\min_{F \in \Span \Delta^*(\delta_T,\delta_X)} N - N_F,\]
so that we aim to bound from below $N-N_F$ uniformly in $F \in \Span \Delta^*(\delta_T,\delta_X)$.

\medskip

Let us fix such a polynomial $F \in \Span \Delta^*(\delta_T,\delta_X)$. 

Second, we aim to regard $N_F$ as the dimension of some vector space. For this purpose, fix $(\epsilon_T,\epsilon_X) \in \Z \times \N$ and consider the map
\[\ev_{(\epsilon_T,\epsilon_X),F}: \left\{ \begin{array}{rcl} R(\epsilon_T,\epsilon_X) & \rightarrow & \F_q^{N_F} \\
G & \mapsto & \left(G(Q)\right)_{Q \in \mathcal{Z}(F)(\F_q)}
\end{array}\right..\]

For $\epsilon_T,\: \epsilon_X \geq q$ the evaluation map $\ev_{(\epsilon_T,\epsilon_X)}$ is surjective by Example \ref{surj}. The map $\ev_{(\epsilon_T,\epsilon_X),F}$  is thus also surjective for any $F \in R(\delta_T,\delta_X)$, as illustrated by the diagram

\begin{figure}[h!]
\centering
\begin{tikzpicture}
\matrix (m) [matrix of math nodes,column sep=6em]{
R(\epsilon_T,\epsilon_X)& \F_q^{\mathcal{H}_\eta(\F_q)} & \F_q^{\mathcal{Z}(F)(\F_q)}.\\
};
\path[->>]
(m-1-1) edge node[below]{$\ev_{(\epsilon_T,\epsilon_X)}$} (m-1-2)
		edge [bend left=20] node[above]{$\ev_{(\epsilon_T,\epsilon_X),F}$} (m-1-3)
(m-1-2) edge (m-1-3);
\end{tikzpicture}
\end{figure}
It follows that
\[N_F=\dim \left(\faktor{R(\epsilon_T,\epsilon_X)}{\ker \ev_{(\epsilon_T,\epsilon_X),F}}\right).\] 

\medskip

Third we aim to display an upper bound $\tilde{N}_F$ of $N_F$ such that $N-\tilde{N}_F$ turns to be easier to handle. Let us denote by $\left\langle F \right\rangle$ the ideal of $R$ generated by $F$ and by $\left\langle F \right\rangle_{(\epsilon_T,\epsilon_X)}$ the subspace $F R(\epsilon_T-\delta_T,\epsilon_X-\delta_X) \subset R(\epsilon_T,\epsilon_X)$ spanned by $F$.

Observing that $\ker \ev_{(\epsilon_T,\epsilon_X)} +\left\langle F \right\rangle_{(\epsilon_T,\epsilon_X)} \subset \ker \ev_{(\epsilon_T,\epsilon_X),F}$, we get $\tilde{N}_F \geq N_F$ with
\[\tilde{N}_F=\dim \left(\faktor{R(\epsilon_T,\epsilon_X)}{\ker \ev_{(\epsilon_T,\epsilon_X)}+ \left\langle F \right\rangle_{(\epsilon_T,\epsilon_X)}}\right).\]
Hence
\begin{equation}\label{bound}d_\eta(\delta_T,\delta_X) \geq \min_ {F \in \Span \Delta^*(\delta_T,\delta_X)}N - \tilde{N}_F\end{equation}
and we are now reduced to bound from below $N-\tilde{N}_F$ uniformly in $F \in \Span \Delta^*(\delta_T,\delta_X)$.

\medskip

Fourth, we now prove that
\begin{equation}\label{chouette}N-\tilde{N}_F\geq \# \Delta^*(\epsilon_T,\epsilon_X)_F.\end{equation}
In fact, we display $\Delta^*(\epsilon_T,\epsilon_X)_F$ as a subfamily of $\Delta^*(\epsilon_T,\epsilon_X)$ which complement would be a spanning family of the vector space $R(\epsilon_T,\epsilon_X)$ modulo the vector subspace ${{\ker \ev_{(\epsilon_T,\epsilon_X)}}+ \left\langle F \right\rangle_{(\epsilon_T,\epsilon_X)}}$.

By Corollary \ref{base} and Lemma \ref{repDiv},
\[\Delta^*(\epsilon_T,\epsilon_X)=\{M, \: M \in \Mon(\epsilon_T,\epsilon_X) \text{ such that } \forall \: g \in \mathcal{G}, \: LT(g) \nmid M\}\]
is a basis of $R(\epsilon_T,\epsilon_X)$ modulo $\mathcal{I}(\epsilon_T,\epsilon_X)$. By Example \ref{surj}, its cardinality equals $N$.

As $F$ is a homogeneous element, the ideal  $\mathcal{I} + \left\langle F \right\rangle$ is homogeneous. Let $\widehat{\mathcal{G}}$ be a Gr\"obner basis of the ideal $\mathcal{I} + \left\langle F \right\rangle$ that contains $\mathcal{G} \cup \{F\}$. Using Proposition \ref{Grob} and restricting on each homogeneous component as in Corollary \ref{base}, the set 
\[\tilde{\Delta}(\epsilon_T,\epsilon_X)=\{M, \: M \in \Mon(\epsilon_T,\epsilon_X) \text{ such that } \forall \: h \in \widehat{\mathcal{G}}, \: LT(h) \nmid M\}\]
is a basis of $R(\epsilon_T,\epsilon_X)$ modulo $\mathcal{I}(\epsilon_T,\epsilon_X)+ \left\langle F \right\rangle_{(\epsilon_T,\epsilon_X)}$ of cardinality $\tilde{N}_F$.

Since $\mathcal{G} \subset \widehat{\mathcal{G}}$ and $F \in \widehat{\mathcal{G}}$, we have $\Delta^*(\epsilon_T,\epsilon_X)_F \subset \Delta^*(\epsilon_T,\epsilon_X) \setminus \tilde{\Delta}(\epsilon_T,\epsilon_X)$, from which (\ref{chouette}) follows.

\medskip

We conclude the proof noticing that $\Delta^*(\epsilon_T,\epsilon_X)_F = \Delta^*(\epsilon_T,\epsilon_X)_{\LT(F)}$ for every polynomial $F$ and using (\ref{bound}) and (\ref{chouette}).
\end{proof}

\subsection{Explicit formulae of the minimum distance}\label{choix}

The previous paragraph gives a lower bound of the minimal distance for any couple $(\epsilon_T,\epsilon_X) \in \N^2$ such $\epsilon_T, \,\epsilon_X \geq q$. We aim to maximize the quantity depending on this couple. \textbf{From now, we set}
\[\epsilon_X= q + \delta_X \text{ and } \epsilon_T = q+ \delta\]
 where as usual $\delta=\delta_T+\eta \delta_X$. The hypotheses for $R(\delta_T,\delta_X)$ not to be zero imply that $\epsilon_T$ and $\epsilon_X$ are greater than $q$. By Theorem \ref{DIST}, one way to compute a lower bound of the minimum distance is to calculate $\#\Delta^*(\epsilon_T,\epsilon_X)_M$ for every monomial $M \in \Delta^*(\delta_T,\delta_X)$ and then minimize the quantity over $\Delta^*(\delta_T,\delta_X)$.

\begin{prop}\label{tomin}
Let $(\alpha_0,\beta_0) \in \mathcal{K}^*(\delta_T,\delta_X)$, defined in Notations \ref{notrep} and \ref{etoile}. Then
\[\#\Delta(\epsilon_T,\epsilon_X)_{M(\alpha_0,\beta_0)}=\max(q-\alpha_0+\mathds{1}_{\{\alpha_0=\delta_X\}},1)\max(q-B(\alpha_0)+\mathds{1}_{\{\beta_0=B(\alpha_0)\}},1)\]
with $B(\alpha_0)=\delta -\eta\alpha_0$.
\end{prop}

\begin{proof}
Set $M= M(\alpha_0, \beta_0)=T_1^{\delta -\eta\alpha_0 - \beta_0} T_2^{\beta_0}X_1^{\delta_X-\alpha_0} X_2^{\alpha_0}$ with $(\alpha_0,\beta_0) \in \mathcal{K}^*(\delta_T,\delta_X)$, that is to say according to Notations \ref{notrep} and \ref {etoile}
\[\begin{array}{rcl}0 \leq \alpha_0 \leq \min(\lfloor A \rfloor,q-1) &\text{or}& (\alpha_0=\delta_X \text{ if } \delta_T \geq 0), \\
0 \leq \beta_0 \leq \min(\delta -\eta\alpha_0,q)-1 & \text{or} &\beta_0=\delta_T+  \eta(\delta_X-\alpha_0).
\end{array}\]

Let $N\in \Delta^*(\epsilon_T,\epsilon_X)$. Write 
\[N= T_1^{\epsilon_T+\eta(\epsilon_X-\alpha)- \beta} T_2^\beta X_1^{\epsilon_X-\alpha} X_2^\alpha\]
with $(\alpha,\beta) \in \mathcal{K}^*(\epsilon_T,\epsilon_X)$. Since $\epsilon_T,\: \epsilon_X \geq q$, then
\[
\begin{array}{rcl}0 \leq \alpha \leq q-1 &\text{or}& \alpha=\epsilon_X, \\
0 \leq \beta \leq q-1 & \text{or} &\beta=\epsilon_T+\eta(\epsilon_X-\alpha).
\end{array}\]
Suppose that $M$ divides $N$. Then
\[\begin{aligned}\alpha_0 &\leq \alpha\\
\delta_X-\alpha_0 &\leq \epsilon_X - \alpha\\
\beta_0 &\leq \beta  \\
\delta -\eta\alpha_0 - \beta_0 &\leq \epsilon_T+\eta(\epsilon_X-\alpha) - \beta\end{aligned} \]
One can rewrite the previous conditions as
\[
\begin{array}{ccccc}
&\alpha_0\leq \alpha \leq q + \alpha_0 &\text{ and }& \beta_0 \leq \beta \leq q + \eta(\epsilon_X - \alpha+\alpha_0) + \beta_0.&\end{array}\]
Since $\alpha \leq \epsilon_X$ and $\alpha_0, \beta_0 \in \N$, both upperbounds are greater than $q-1$. Moreover, 
\begin{align*}q + \alpha_0 = \epsilon_X \: &\Leftrightarrow \:  \alpha_0 = \delta_X, \\
q+\eta(\epsilon_X-\alpha+\alpha_0)+\beta_0 = \epsilon_T+\eta(\epsilon_X-\alpha) \: &\Leftrightarrow \: \beta_0=B(\alpha_0),\end{align*}
which justifies the choice of $\epsilon_T$ and $\epsilon_X$ to maximize the quantity $\#\Delta^*(\epsilon_T,\epsilon_X)_{(M(\alpha_0,\beta_0)}$.

\smallskip

To sum up, determining $\# \Delta^*(\epsilon_T,\epsilon_X)_M$ is equivalent to compute the number of couples $(\alpha,\beta)\in \mathcal{K}^*(\epsilon_T,\epsilon_X)$ satisfying the following conditions.
\begin{equation}\tag{$\star$}\label{cond}
\left\{\begin{array}{rcl}
\alpha_0\leq \alpha \leq q-1 &\text{or}& \{\alpha=\epsilon_X \text{ and } \alpha_0=\delta_X\}\\
\beta_0 \leq \beta \leq q-1 &\text{or}& \{\beta=\epsilon_T+\eta(\epsilon_X-\alpha) \text{ and } \beta_0=B(\alpha_0)\}\end{array}\right.\end{equation}

Moreover,
\begin{itemize}
\item If $\delta_X \geq q$ and $\alpha_0=\delta_X$, the only $\alpha$ that satifies (\ref{cond}) is $\alpha=\epsilon_X$.
\item If $\delta_T+\eta\delta_T \geq q$ and $\beta_0=\delta_T+\eta\delta_T$, the only $\beta$ satifying (\ref{cond}) is $\beta=\epsilon_T+\eta\epsilon_X$. 
\end{itemize}
Then, one can write

\[\#\Delta(\epsilon_T,\epsilon_X)_M=\left\{\begin{array}{cl}
(q-\alpha_0)(q-\beta_0) & \text{if } \alpha_0 < \delta_X \text{ and }\beta_0 < B(\alpha_0),\\
(q-\alpha_0)\max(q-B(\alpha_0)+1,1) & \text{if } \alpha_0 < \delta_X \text{ and }\beta_0 =B(\alpha_0),\\
\max(q-\delta_X+1,1)(q-\beta_0) & \text{if } \alpha_0=\delta_X \text{ and } \beta_0 < B(\alpha_0),\\
\max(q-\delta_X+1,1)\max(q-B(\alpha_0)+1,1) & \text{if } \alpha_0=\delta_X \text{ and } \beta_0=B(\alpha_0).
\end{array}\right.\]

Let us highlight that a couple $(\alpha_0,\beta_0) \in \mathcal{K}^*(\delta_T,\delta_X)$ is less or equal to $q-1$ or equal to $\delta_X$. Then either $(\alpha_0,\beta_0)$ or $(\alpha_0,\epsilon_T+\eta(\epsilon_X-\alpha_0))$ satisfies (\ref{cond}). Then the quantity $\#\Delta^*(\epsilon_T,\epsilon_X)_{M(\alpha_0,\beta_0)}$ can never be zero.
\end{proof}

\smallskip

To lowerbound the minimum distance, it remains to minimize
\[\max(q-\alpha_0+\mathds{1}_{\{\alpha_0=\delta_X\}},1)\max(q-\beta_0+\mathds{1}_{\{\beta_0=B(\alpha_0)\}},1)\]
over $(\alpha_0,\beta_0) \in \mathcal{K}^*(\delta_T,\delta_X)$. The problem can be reduced to minimize a univariate function, thanks to the following lemma.

\begin{lemma}\label{rewrite}
Let $\eta \geq 0$ and $(\delta_T,\delta_X) \in \Z \times \N$ such that $\delta \geq 0$. Then
\[d_\eta(\delta_T,\delta_X) \geq \min_{\alpha_0\in \mathcal{A}^*_X} f(\alpha_0)\]
with 
\[\mathcal{A}^*_X=\left\{\begin{array}{cl}
\lsem 0 ,\max(q,\delta_X) -1 \rsem \cup \{\delta_X\} &\text{if } \delta_T \geq 0,\\
\left\lsem 0 ,\max\left(q-1,\left\lfloor \frac{\delta}{\eta}\right\rfloor\right)\right\rsem & \text{if } \delta_T < 0\end{array}\right.\]
 and
\[f(\alpha_0)=\max(q-\alpha_0 + \mathds{1}_{\alpha_0=\delta_X},1)\max(q-\delta +\eta\alpha_0 + 1,1)\]
\end{lemma}

\begin{proof}
By Theorem \ref{DIST}, we have to minimize $\#\Delta^*(\delta_T,\delta_X)_{M_{(\alpha_0,\beta_0)}}$ for $(\alpha_0,\beta_0) \in \mathcal{K}^*(\delta_T,\delta_X)$.
The only observation we need to prove this lemma is that for each $\alpha_0 \in \mathcal{A}^*_X$, for all $\beta \in \lsem 0 , \min(B(\alpha_0),q)-1\rsem$, 
\[q-\beta_0 \geq q-B(\alpha_0) + 1.\]
Substituting in the formula of Proposition \ref{tomin} gives the desired conclusion.
\end{proof}

In other words, Lemma \ref{rewrite} means that the minimum is reached by monomials of the form $M(\alpha_0,\delta -\eta\alpha_0)$ for $\alpha_0 \in \mathcal{A}^*_X$.

\begin{prop}\label{distance}
Let $\eta \geq 0$, $(\delta_T,\delta_X) \in \Z \times \N$ with $\delta \geq 0$. The code $C_\eta(\delta_T,\delta_X)$ on the Hirzebruch surface $\mathcal{H}_\eta$ has minimum distance that is given as follows:
\begin{itemize}
\item If $\eta \geq 2$,
\begin{itemize}
\item If $q > \delta$, then
\[d_\eta(\delta_T,\delta_X)= (q+\mathds{1}_{\delta_X=0})(q-\delta+1),\]
\item If $\max\left( \frac{\delta}{\eta +1},\delta_T\right) < q \leq \delta$, then 
\[d_\eta(\delta_T,\delta_X)=q - \left\lfloor \frac{\delta-q}{\eta} \right\rfloor,\]
\item If $ q \leq \max\left( \frac{\delta}{\eta +1},\delta_T\right)$,
\[d_\eta(\delta_T,\delta_X)= \left\{\begin{array}{cl}\max(q-\delta_X+1,1) & \text{if } \delta_T \geq 0,\\
1 &\text{if } \delta_T <0,
\end{array}\right.\]
\end{itemize} 
\item if $\eta=0$,
\[d_\eta(\delta_T,\delta_X) = \max(q-\delta_X + 1, 1)  \max(q-\delta_T + 1, 1).\]
\end{itemize}
\end{prop}

We first prove the lower bounding. Equality will follow from Proposition \ref{atteint}, exhibiting polynomials associated to words of minimum weight.

\begin{proof}

By Lemma \ref{rewrite}, we aim to prove the lower bound by minimizing the function $f$ on $\mathcal{A}^*_X$, the function and the set depending on parameters $\eta,\: \delta_T, \: \delta_X$ and $q$.

Let us highlight that $\mathcal{A}^*_X \subset [0,\delta_X]$ regardless of parameters.

The form of the function $f$ as a product of two maxima of a linear function with $1$  implies that the real function $f : [0,\delta_X] \rightarrow \R$ is a concave piecewise function. The pieces depend on the size of $q$ with respect to the parameters.

More precisely, note that
\[q-\delta +\eta\alpha_0 + 1 \leq 1 \: \Leftrightarrow \: \alpha_0 \leq s \]
where $\displaystyle{s=\frac{\delta-q}{\eta}}$ has already been defined in Section \ref{sectionS}.

Then $f$ is a piecewise function that has a decreasing linear polynomial on the interval $[0,s]$ and a concave quadratic function on the interval $[s,\delta_X[$ with negative dominant coefficient, provided that $s \in [0,\delta_X[$.

If $s \leq 0$, then the function $f$ is quadratic and concave on $[0,\delta_X[$. Finally, if $s \geq \delta_X$, then $f$ is decreasing on $[0,\delta_X[$.

Then the minimum point of $f$ on $\mathcal{A}^*_X$ is the floor or the ceiling of  one the bound of these intervals. 

\textbf{Let us first suppose that $\eta  \geq 2$ and $\delta_T \geq 0$.} Then
\[\mathcal{A}^*_X=\left\{\begin{array}{cl} \lsem 0,\delta_X \rsem & \text{if } \delta_X \leq q \\
 \lsem 0,q-1 \rsem \cup \{\delta_X\} & \text{if } \delta_X \geq q \end{array} \right.\]

\begin{enumerate}
 \item If $\delta< q$, then $s < 0$ and $ q > \delta_X$. In this case 
\[f(\alpha_0)=\left\{\begin{array}{cl}
(q-\alpha_0)(q-\delta +\eta\alpha_0 +1) &\text{if } \alpha_0 \neq \delta_X, \\
(q-\delta_X+1)(q-\delta_T+1) & \text{if } \alpha_0=\delta_X.
\end{array}\right.\]

 \begin{itemize}
\item If $\delta_X=0$, then $\mathcal{A}^*_X=\{0\}$ and 
 \[\min_{\alpha_0 \in \mathcal{A}_X}f(\alpha_0)= f(\delta_X)=(q+1)(q-\delta_T+1).\] 
 
\item Otherwise, on the interval $[0,\delta_X-1]$, the minimum is reached by one of the bounds of the interval, i.e. $\alpha_0=0$ or $\alpha_0=\delta_X-1$ (see Fig. \ref{f8}).

\begin{figure}[h]
\centering
\begin{tikzpicture}
		\def\dX{5};
		\def\dT{1};
		\def\n{1};
		\def\q{7};
		\def\imX{(\q-\dX+1)*(\q-\dT+1)}
		\pgfmathsetmacro{\im}{max(\q-\dX+1,1)*max(\q-\dT+1,1)}
 		\pgfmathsetmacro{\a}{\dX-1} 
	    \begin{axis}[xmin=0,ymin=10,clip=true,axis lines*=middle, grid = both,
    grid style={dashed}, xtick distance=1, ytick distance=1,x=0.5cm,y=0.5cm]
		\addplot[domain=0:\a,red] {(\q-x)*(\q-\dT-\n*(\dX-x)+1)};
		\addplot[color=black,mark=*] coordinates{(\dX,\imX)};
		\addplot[color=black,mark=*] coordinates{(4,18)};
		\addplot[color=black,mark=*] coordinates{(0,14)};
		\end{axis}
		\draw [dashed] (2.5,0) node[below,fill=white,font=\boldmath]{$\delta_X$} -- (2.5,5.5);
    \end{tikzpicture} 
   \caption{Graph of $f$ when $q> \delta$ \\ e.g. $(\eta,\delta_T,\delta_X,q)=(1,1,5,7)$}\label{f8}
   \end{figure}
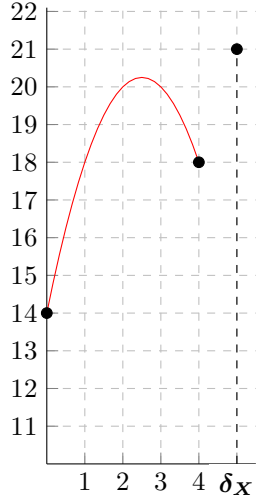

In addition, one can notice that $f(\delta_X)=f(\delta_X-1)+\eta (q-\delta_X-1)$. Then
\[f(\delta_X) > f(\delta_X-1).\]
The minimum of the function $f$ on $\mathcal{A}^*_X$ is reached either by $\alpha_0=0$ or $\alpha_0=\delta_X-1$. It remains to compare both values.

We have $f(0) \leq f(\delta_X-1)$ if and only if
\[q\eta\delta_X \geq (\delta_X-1)(q-\delta_T-\eta+1)+\eta q\]
which is equivalent to
\begin{equation}\label{tjrs}q(\delta_X-1)(\eta-1) \geq -(\delta_X-1)(\delta_T+\eta-1)\end{equation}
Since $\eta \geq 2$, $\delta_T \geq 0$, $\delta_X \geq 1$ and $q \geq 2$, the left hand side is non negative, whereas the the right hand side is non positive. The inequation (\ref{tjrs}) is thus always satisfied and 
\[\min_{\alpha_0 \in \mathcal{A}^*_X}f(\alpha_0)= f(0)= q(q-\delta+1)\]
\end{itemize}  
 
 \item If $\max\left( \frac{\delta}{\eta +1},\delta_T\right) < q \leq \delta$, then $\delta_X \geq 1$ and $\lfloor s \rfloor \in \lsem 0,\min(\delta_X,q)-1\rsem$). In this case
 \[f(\alpha_0)=\left\{\begin{array}{cl}
(q-\alpha_0) &\text{if } \alpha_0 \leq s,\\
(q-\alpha_0)(\eta(\alpha_0-s) +1) &\text{if } \alpha_0 \neq \delta_X \text{ and } \alpha_0 \geq s,\\
\max(q-\delta_X+1,1)(q-\delta_T+1) & \text{if } \alpha_0=\delta_X.
\end{array}\right.\]
See Figure \ref{f9} for examples of graph of the function $f$.
\begin{figure}[h!]

\begin{subfigure}[b]{0.48\textwidth}
\centering
\begin{tikzpicture}
		\def\dX{9};
		\def\dT{1};
		\def\n{3};
		\def\q{9};
		\pgfmathsetmacro{\im}{max(\q-\dX+1,1)*max(\q-\dT+1,1)}
		\def\A{ \dX < \q ? \dX : \q}
 		\pgfmathsetmacro{\a}{\A-1} 
		\def\S{(\dT+\n*\dX-\q)/\n}
		\pgfmathsetmacro{\s}{ \S < 0 ? 0 : \S > \a ? \a : \S} 
		\def\ss{floor(\s)}
		\def\ims{(\q-\ss-1)*(\q-\dT-\n*(\dX-\ss-1)+1)}
	    \begin{axis}[xmin=0,ymin=0,clip=false,axis lines*=middle, grid = both,xtick={1,2,3,4,5,8},
    grid style={dashed}, xtick distance=1, ytick distance=1,x=0.5cm,y=0.5cm]
	    \addplot[domain=0:\s,blue] {\q-x};
		\addplot[domain=\s:8,red] {(\q-x)*(\q-\dT-\n*(\dX-x)+1)};
		\addplot[color=black,mark=*] coordinates{(\ss,\q-\ss)};
		\addplot[color=black,mark=*] coordinates{(\ss+1,\ims)};
		\addplot[color=black,mark=*] coordinates{(\dX,\im)};
		\addplot[color=black,mark=*] coordinates{(8,6)};
		\end{axis}
		\draw [dashed] (4.5,0) node[below=0.05cm,font=\boldmath]{$\delta_X$} -- (4.5,4.5);
		\draw [dashed] (3,0) node[below=0.05cm,font=\boldmath]{$\lfloor s \rfloor$} -- (3,1.5);
		\draw [dashed] (3.5,0) node[below=0.05cm,font=\boldmath]{$\lceil s \rceil$} -- (3.5,3);
    \end{tikzpicture}
  \caption{$\delta_X \geq q$ \\
  e.g. $(\eta,\delta_T,\delta_X,q)=(3,1,9,9)$}
   \end{subfigure} \hfill
   \begin{subfigure}[b]{.48\linewidth}
	\centering
\begin{tikzpicture}
		\def\dX{3};
		\def\dT{3};
		\def\n{1};
		\def\q{5};
		\pgfmathsetmacro{\im}{max(\q-\dX+1,1)*max(\q-\dT+1,1)}
		\def\A{ \dX < \q ? \dX : \q}
 		\pgfmathsetmacro{\a}{\A-1} 
		\def\S{(\dT+\n*\dX-\q)/\n}
		\pgfmathsetmacro{\s}{ \S < 0 ? 0 : \S > \a ? \a : \S} 
		\def\ss{floor(\s)}
		\def\ims{(\q-\ss-1)*(\q-\dT-\n*(\dX-\ss-1)+1)}
	    \begin{axis}[xmin=0,ymin=0,clip=false,axis lines*=middle, grid = both, xtick={2},
    grid style={dashed}, xtick distance=1, ytick distance=1,x=0.5cm,y=0.5cm]
	    \addplot[domain=0:\ss,blue] {\q-x};
		\addplot[domain=\s:2,red] {(\q-x)*(\q-\dT-\n*(\dX-x)+1)};
		\addplot[color=black,mark=*] coordinates{(\ss,\q-\ss)};
		\addplot[color=black,mark=*] coordinates{(2,6)};
		\addplot[color=black,mark=*] coordinates{(\dX,\im)};
		\end{axis}
		\draw [dashed] (1.5,0) node[below,font=\boldmath]{$\delta_X$} -- (1.5,4.5);
		\draw [dashed] (0.5,0) node[below,font=\boldmath]{$s$} -- (0.5,2);
    \end{tikzpicture}
  \caption{$\delta_X < q$ \\ e.g. $(\eta,\delta_T,\delta_X,q)=(1,3,3,5)$}
 \end{subfigure} 
 \caption{Examples of graph of the function $f$ when $\max\left( \frac{\delta}{\eta +1},\delta_T\right) < q \leq \delta$}\label{f9}
 \end{figure}
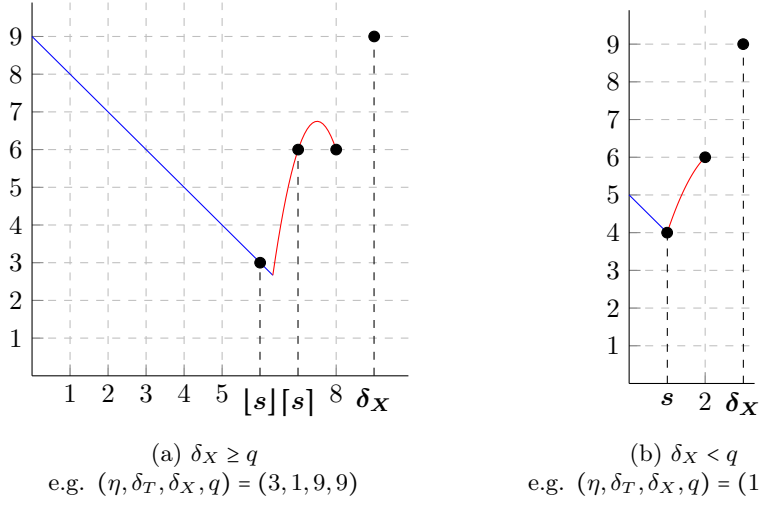

The possible arguments for the minimum are $\lfloor s \rfloor$, $\lfloor s \rfloor+1$, $\min(q,\delta_X)-1$ and $\delta_X$.

First notice that  $\lfloor s \rfloor \leq q-1$ and
\[\begin{aligned}f(\lfloor s \rfloor +1) &= (q-\lfloor s \rfloor-1)(\eta(\lfloor s \rfloor+1-s)+1\\
&= f(\lfloor s \rfloor) -1 + (q-\lfloor s \rfloor)\eta(\lfloor s \rfloor+1-s).\end{aligned}\]
Therefore
\begin{equation}\label{s}f(\lfloor s \rfloor+1) \geq f(\lfloor s \rfloor).\end{equation}

Second let us check that the minimum of $f$ cannot be reached by $\alpha_0=\delta_X$.
\begin{itemize}
\item If $\delta_X \leq q$, then $f(\delta_X)=f(\delta_X-1)+\eta (q-\delta_X-1)$ and $f(\delta_X) > f(\delta_X-1)$.
\item If $\delta_X \geq q$, then $f(\delta_X)=\eta(\delta_X-s)+1 \geq \eta(q-1-s)+1 = f(q-1)$.
\end{itemize}

Then the minimum of $f$ is reached by either $\alpha_0=\lfloor s \rfloor$ or $\alpha_0=\min(\delta_X,q)-1$.

\begin{itemize}
\item If $\delta_X \geq q$, we want to prove that $f(\lfloor s \rfloor) \leq f(q-1)$.
\[\begin{aligned}f(q-1)&=\eta(q-s-1)+1& \\
&\geq \eta(q-\lfloor s \rfloor -1)+ 1& \\
&\geq q-\lfloor s \rfloor +(\eta-1)(q-\lfloor s \rfloor -1)&\text{since } \lfloor s \rfloor \leq q-1\\
& \geq q-\lfloor s \rfloor =f(\lfloor s \rfloor)\end{aligned}\]

\item If  If $\delta_X \leq q$,  we want to prove that $f(\lfloor s \rfloor) \leq f(\delta_X-1)$. 

Let us assume $\lfloor s \rfloor \neq \delta_X-1$ and $f(\lfloor s \rfloor) > f(\delta_X-1)$ and let us display a contradiction.
\[f(\lfloor s \rfloor) > f(\delta_X-1) \: \Leftrightarrow \: \lfloor s \rfloor < \delta_X-1 -\eta(\delta_X-1-s)(q-\delta_X+1)\]
Since the right hand side is an integer, we have
\[f(\lfloor s \rfloor) > f(\delta_X-1) \: \Rightarrow \: s < \delta_X-1 -\eta(\delta_X-1-s)(q-\delta_X+1)\]
Replacing $s$ by its value, we get
\[\frac{\delta_T}{\eta} < - 1 - \eta(\delta_X-1-s)(q-\delta_X+1)\]
But the assumption $\lfloor s \rfloor \neq \delta_X-1$ ensures that $\delta_X-1 > s$ and then $\eta(\delta_X-1-s)(q-\delta_X+1) \geq 0$. The right handside being negative, it is a contradiction with $\delta_T \geq 0$.
\end{itemize}
  
Then, in both cases,
\[\min_{\alpha_0 \in \mathcal{A}^*_X}f(\alpha_0)= f(\lfloor s \rfloor)= q-\lfloor s \rfloor\]

 \item if $q \leq \max\left( \frac{\delta}{\eta +1},\delta_T\right)$, then $s \geq \min(\delta_X,q)$ and
 \[f(\alpha_0)=\max(q-\alpha_0+\mathds{1}_{\alpha_0=\delta_X},1).\]
This is a decreasing function on $[0,\delta_X]$, as $f(\delta_X-1)=f(\delta_X)$. It follows easily that
 \[\min_{\alpha_0 \in \mathcal{A}^*_X}f(\alpha_0)= f(\delta_X)=\left\{ \begin{array}{cl}
1 &\text{if } q < \delta_X,  \\
 q-\delta_X+1 &\text{if } q \geq \delta_X.
\end{array}\right.\]  
 \end{enumerate}


\textbf{Now, let us focus on the case $\eta \geq 2$ and $\delta_T <0$.} Let us recall that $A=\frac{\delta}{\eta} < \delta_X$ does not belong to $\mathcal{A}^*_X$ if $A \geq q$ and 
\[\mathcal{A}^*_X=\left\{\begin{array}{cl} \lsem 0,\lfloor A \rfloor \rsem & \text{if } A < q, \\
 \lsem 0,q-1 \rsem & \text{if } A \geq q. \end{array} \right.\]
Moreover, $s=A - \frac{q}{\eta} < A$. . Since $\delta_X \notin \mathcal{A}^*_X$, one can rewrite
\[f(\alpha_0)=\left\{\begin{array}{cl}
(q-\alpha_0) &\text{if } \alpha_0 \leq s,\\
(q-\alpha_0)(q-\eta(A-\alpha_0) +1) &\text{if } \alpha_0 \geq s\end{array}\right.\]

\begin{enumerate}
 \item If $\delta=\eta A < q$, then $s<0$ and $\mathcal{A}^*_X=\lsem 0, \lfloor A \rfloor \rsem$. Therefore the function $f$ can be written
 \[f(\alpha_0)=(q-\alpha_0)(q+\eta(\alpha_0-A)+1)\]
It is increasing then decreasing on $\mathcal{A}^*_X$ so its minimum is reached for either $\alpha_0=0$ or $\alpha_0=\lfloor A \rfloor$. Let us compare $f(0)$ and $f(\lfloor A \rfloor)$.

\[f(0) \leq f(\lfloor A \rfloor) \: \Leftrightarrow \: \lfloor A \rfloor(q+\eta (\lfloor A \rfloor-A) + 1) \leq q \eta \lfloor A \rfloor\]

If $\lfloor A \rfloor \neq 0$ we can simplify by $\lfloor A \rfloor$ and, writing $\{A\}=A - \lfloor A \rfloor $, we get
\[f(0) \leq f(\lfloor A \rfloor) \: \Leftrightarrow \: 1 -\eta \{A \} \leq q (\eta -1)\]
However, $0 \leq \eta \{A\} \leq \eta-1$, which implies that $1-\eta\{A\} \leq 1$ whereas the right hand-side is a non negative integer. Then the right hand-side is greater than the left one if and only if it is non zero, which is equivalent to to $\eta \geq 2$, which is always true\footnote{Here is one of the arguments that fail when $\eta=1$.}.

Otherwise, it is obvious. Then
\[\min_{\alpha_0 \in \mathcal{A}^*_X}f(\alpha_0) =  f(0)=q(q-\delta+1). \]

 \item If $\frac{\delta}{\eta+1} < q \leq \delta$, then
 \[\lfloor s \rfloor \in \mathcal{A}^*_x = \left\{\begin{array}{cl}
 \lsem 0,\lfloor A \rfloor \rsem &\text{if } q > A = \frac{\delta}{\eta},\\
 \lsem 0 ,q-1 \rsem &\text{if } q \leq A.
 \end{array}\right.\]
 
 The function $f$ has a linear decreasing piece on $[ 0, s ]$ and then it is concave on $[0,\min(\lfloor A \rfloor,q-1)]$, as illustrated in Figure \ref{f9}. 
 
    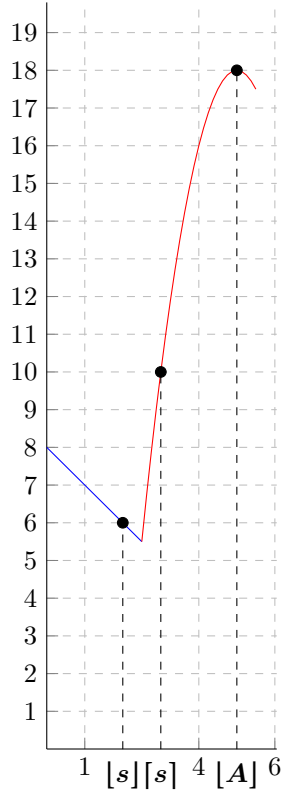
\begin{figure}[h] 
\centering
\begin{tikzpicture}
		\def\dX{9};
		\def\dT{-5};
		\def\n{2};
		\def\q{8};
		\pgfmathsetmacro{\im}{max(\q-\dX+1,1)*max(\q-\dT+1,1)}
		\def\A{ \dX+\dT/\n < \q ?  \dX+\dT/\n : \q}
 		\pgfmathsetmacro{\a}{\A-1} 
		\def\S{(\dT+\n*\dX-\q)/\n}
		\pgfmathsetmacro{\s}{ \S < 0 ? 0 : \S > \a ? \a : \S} 
		\def\ss{floor(\s)}
		\def\ims{(\q-\ss-1)*(\q-\dT-\n*(\dX-\ss-1)+1)}
	    \begin{axis}[xmin=0,ymin=0,clip=false,axis lines*=middle, grid = both, xtick={1,4,6},
    grid style={dashed}, xtick distance=1, ytick distance=1,x=0.5cm,y=0.5cm]
	    \addplot[domain=0:\s,blue] {\q-x};
		\addplot[domain=\s:\a-1,red] {(\q-x)*(\q-\dT-\n*(\dX-x)+1)};
		\addplot[color=black,mark=*] coordinates{(\ss,\q-\ss)};
		\addplot[color=black,mark=*] coordinates{(\ss+1,\ims)};
		\addplot[color=black,mark=*] coordinates{(5,18)};
		\end{axis}
		\draw [dashed] (1,0) node[below =0.05cm,font=\boldmath]{$\lfloor s \rfloor$} -- (1,3);
		\draw [dashed] (1.5,0) node[below=0.05cm,font=\boldmath]{$\lceil s \rceil$} -- (1.5,5);
		\draw [dashed] (2.5,0) node[below=0.05cm,font=\boldmath]{$\lfloor A \rfloor$} -- (2.5,9);
    \end{tikzpicture} 
   \caption{Graph of $f$ for $\frac{\delta}{\eta} < q \leq \delta$ and $\eta \geq 2$ \\ e.g. $(\eta,\delta_T,\delta_X,q)=(2,-5,9,8)$}\label{f10}
   \end{figure}

Then it can be proved in a same way as in the second case for $\delta_T \geq 0$ (Equation (\ref{s})) that
 \[f(\lfloor s \rfloor+1) \geq f(\lfloor s \rfloor).\]
The minimum of $f$ on $\mathcal{A}^*_X$ is thus either reached for $\alpha_0=\lfloor s \rfloor$ or
\begin{itemize}
\item $\alpha_0=\lfloor A \rfloor$ if $q > A$,
\item $\alpha_0=q-1$ if $q \leq A$.
\end{itemize}

Let us prove that the minimum is reached at $\alpha_0=\lfloor s \rfloor$ in both cases.
\begin{itemize}
\item If $q > A$, let us first notice that, since $s < A$, we have $\lfloor s \rfloor \leq \lfloor A \rfloor$.

If they are equal, the problem is solved.

 Otherwise, one can write 
 \[f(A)=(q-A)(q+1) > f(\lfloor s \rfloor).\] As a fonction on $\R$, $f$ is increasing on $[ \lfloor s \rfloor+1, A]$, then 
 \[f(\lfloor A \rfloor) \geq f(\lfloor s +1 \rfloor) \geq  f(\lfloor s \rfloor).\]
 \item If $q \leq A$, we have 
 \[f(q-1)=\eta(q-1-s)+1=q-s +(\eta-1)(q-1-s).\]
 Since $\eta \geq 2$, we have
\[\begin{aligned} f(q-1)&=q-s +(\eta-1)(q-1-s),&&&\\
&\geq q-s +(\eta-1)\left(q\left(1-\frac{1}{\eta}\right)+1\right)&&\text{since } \frac{\eta-1}{\eta}q \leq s,\\
& \geq q-s + \left(\frac{q}{2} + 1\right) &&\text{because } \eta \geq 2,\\
& \geq q-s+1 &&\text{because } q \geq 2,\\
& \geq f(\lfloor s \rfloor)= q-\lfloor s \rfloor&& \text{because } s-1 \leq \lfloor s \rfloor.\end{aligned}\]

\end{itemize}
\begin{flalign*}\text{Then, in both cases, } &\min_{\alpha_0 \in \mathcal{A}^*_X}f(\alpha_0)= f(\lfloor s \rfloor)=q - \lfloor s \rfloor. & \end{flalign*}
  
 \item If $q \leq \frac{\delta}{\eta+1}$, then $s \geq q$. For all $\alpha \in \mathcal{A}^*_X= \lsem 0 ,q-1 \rsem$, 
\[f(\alpha_0)=\max(q-\alpha_0,1).\] 
 Since $q-1 \in \mathcal{A}^*_X$, we have
  \[\min_{\alpha_0 \in \mathcal{A}^*_X}f(\alpha_0) = f(q-1)=1\]
 \end{enumerate}

\textbf{Finally, for $\eta=0$}, the expression in the first maximum is a decreasing function of $\alpha_0$ and the expression in the second maximum does not depend on $\alpha_0$ anymore. Then
\[\min_{\alpha_0 \in \mathcal{A}^*_X}f(\alpha_0) = f(\delta_X)=\max(q-\delta_X + 1, 1)  \max(q-\delta_T + 1, 1)\]

\end{proof}

The following proposition displays some polynomials which codewords has weight that reaches the lower bound given in Proposition \ref{distance}.

\begin{prop}\label{atteint}
Write $\F_q=\{\xi_1,\xi_2,\dots,\xi_q\}$.
\begin{itemize}
\item If $\eta \geq 2$,
\begin{itemize}
\item If $q > \delta$, set
\[F(T_1,T_2,X_1,X_2)=
X_1^{\delta_X} \prod_{i=1}^{\delta} (T_2-\xi_i T_1), \]
\item If $\max\left( \frac{\delta}{\eta +1},\delta_T\right) < q \leq \delta$, write $s=\lfloor \frac{\delta - q}{\eta}\rfloor$. Set
\[F(T_1,T_2,X_1,X_2)=T_2^{\delta-\eta s-q}(T_2^q-T_2T_1^{q-1})X_1^{\delta_X-s}\prod_{i=1}^{s} (X_2-\xi_i T_1^\eta X_1),\]

\item If $ q \leq \max\left( \frac{\delta}{\eta +1},\delta_T\right)$, set
\[F(T_1,T_2,X_1,X_2)=\left\{\begin{array} {cl}
T_2^{\delta_T-q}(T_2^q -T_2T_1^{q-1}) \prod_{i=1}^{\delta_X}(X_2-\xi_i X_1 T_1 ^\eta)& \text{if } \delta_X< q,\\
X_2^{\delta_X-q}T_2^{\delta-q}(T_2^q -T_2T_1^{q-1})\prod_{i=1}^{q}(X_2-\xi_i X_1 T_1^\eta)&\text{if } \delta_X \geq q,\end{array}\right.\]
\end{itemize} 
\item if $\eta=0$, set $m_T=\min(q,\delta_T)$ and $m_X=\min(q,\delta_X)$. Set
\[F(T_1,T_2,X_1,X_2)=X_2^{\delta_X-m_X} T_2^{\delta_T-m_T} \prod_{i=1}^{m_X}(X_2-\xi_i X_1 T_1 ^\eta) \prod_{j=1}^{m_T}(T_2-\xi_j T_1).\]
\end{itemize}
Then the weight of the codeword associated to $F$ in $C_\eta(\delta_T,\delta_X)$ reaches the minimum distance.
\end{prop}

\begin{remark}
\begin{enumerate}
\item The minimum $\#\Delta^*(\delta_T,\delta_X)_{M}$ on $M \in \Delta^*(\delta_T,\delta_X)$ is reached for by the leading term of $F$ in each case. 
\item The previous proposition guarantees us than the choice of $\epsilon_T$ and $\epsilon_X$ in Paragraph \ref{choix} is adequate.
\item Focusing on the points lying on the torus, J. Little and H. Schenck \cite{Little} already proved that the polynomial with the most zero $\F_q$-points on a Hirzeburch surface have the form given in Proposition \label{atteint} for $q$ large enough to make the evaluation map injective.
I. Soprunov and E. Soprunova \cite{SS} demonstrated that the number of $\F_q$ torus-points of a curve defined by $f=0$ depends on the number $L$ of absolutely irreducible factors of $f$:
\[\#C(\F_q^\times)\leq L(q-1) + \lfloor 2 \sqrt{q}\rfloor -1.\]
Even if one could fairly excepted that maximal curves are union of ``lines'', a comprehensive computation of polynomials associated to minimal codewords highlights non linear factors among these polynomials, as stated by I. Soprunov and E. Soprunova.
\end{enumerate}
\end{remark}

\begin{proof}
First, suppose $\eta=2$.
\begin{itemize}
\item If $q > \delta$, the polynomial $F(T_1,T_2,X_1,X_2)=
X_1^{\delta_X} \prod_{i=1}^{\delta} (T_2-\xi_i T_1)$ vanishes at every point of the form $(1,\xi_i,x_1,x_2)$ or $(t_1,t_2,0,1)$, that is to say at $(\delta)(q+1) + q+1- \delta$ points.
\item If $\max\left( \frac{\delta}{\eta +1},\delta_T\right) < q \leq \delta$, note that 
$T_2^q-T_2T_1^{q-1}=\prod_{a \in \F_q}(T_2-a T_1)$. Then the polynomial
\[F(T_1,T_2,X_1,X_2)=T_2^{\delta-\eta s-q}\prod_{a \in \F_q}(T_2-a T_1) X_1^{\delta_X-s}\prod_{i=1}^{s} (X_2-\xi_i T_1^\eta X_1) \]
vanishes at every point except at the ones of the form $(0,1,1,\xi)$ with $\xi \notin \{\xi_i, i \in \{1,\dots,s\}\}$. The code word associated has weight equal to $q-s$.
\item Assume $ q \leq \max\left( \frac{\delta}{\eta +1},\delta_T\right)$.
If $q > \delta_X$, the polynomial
\[F(T_1,T_2,X_1,X_2)=T_2^{\delta_T-q}(T_2^q -T_2T_1^{q-1}) \prod_{i=1}^{\delta_X}(X_2-\xi_i X_1 T_1 ^\eta)\]
vanishes at the point with the form $(1,a,x_1,x_2)$ and $(0,1,1,\xi_i)$, that is to say $q(q+1)+\delta_X$. 

If $\delta_X \geq q$, the only point at which $F= X_2^{\delta_X-q}T_2^{\delta-q}(T_2^q -T_2T_1^{q-1})\prod_{i=1}^{q}(X_2-\xi_i X_1 T_1^\eta)$ is not zero is $(0,1,0,1)$.
\end{itemize} 
Finally, if $\eta=0$, the polynomial 
\[F(T_1,T_2,X_1,X_2)=X_2^{\delta_X-m_X} T_2^{\delta_T-m_T} \prod_{i=1}^{m_X}(X_2-\xi_i X_1 T_1 ^\eta) \prod_{j=1}^{m_T}(T_2-\xi_j T_1)\]
vanishes at every point of the form $(t_1,t_2,1,\xi_i)$ and $(1,\xi_j,x_1,x_2)$, i.e. at $(m_T+m_X)(q+1)-m_Tm_X$ points. Moreover, if $q < \delta_X$ (resp. $q < \delta_T$), it also vanishes at $(t_1,t_2,1,0)$ (resp. $(1,0,x_1,x_2)$).
\end{proof}

\begin{remark}
The parameters for the code $C_0(\delta_T,\delta_X)$ on $\PP^1 \times \PP^1$, are the same as in \cite{CD} (see Theorem 2.1 and Remark 2.2).
\end{remark}

\section{Upperbound on the number of $\F_q$-rational points of curves on Hirzebruch surfaces}\label{nombre}\label{seccurve}

Proposition \ref{distance} gives an upper bound on the number of $\F_q$-rational points of a non-filling curve on a Hirzebruch surface $\mathcal{H}_\eta$. It is worth to highlight that there exists a filling curve of bidegree $(\delta_T,\delta_X)$ if and only if $q < \delta$.

\begin{corollary}
Let $\eta \geq 0$, $\eta \neq 1$ and $(\delta_T,\delta_X) \in \Z \times \N$ with $\delta=\delta_T+\eta \delta_X \geq 0$. Let $\mathcal{C}$  be a non-filling curve on the Hirzebruch surface $\mathcal{H}_\eta$ which Picard class is $\delta_T \mathcal{F} + \delta_X \sigma$. Then the number of $\F_q$-rational point of the curve $\mathcal{C}$ is upper-bounded as follow.

\begin{itemize}
\item If $\eta \geq 2$,
\begin{itemize}
\item If $q > \delta$, then
\[\#\mathcal{C}(\F_q)\leq\left\{ \begin{array}{cl}
(q+1)\delta_T & \text{if } \delta_X=0 \text{ and } \delta_T \geq 0 ,\\
q(\delta+1)+1  &  \text{otherwise.}\end{array}\right.\]
\item If $\max\left( \frac{\delta}{\eta +1},\delta_T\right) < q \leq \delta$, then 
\[\#\mathcal{C}(\F_q)\leq q^2+q+1+ \left\lfloor \frac{\delta-q}{\eta} \right\rfloor. \]
\item If $ q \leq \max\left( \frac{\delta}{\eta +1},\delta_T\right)$ and $q \geq \delta_X$,
\[\#\mathcal{C}(\F_q)\leq q^2+q+\delta_X. \]
\end{itemize}
\item if $\eta=0$,
\[\#\mathcal{C}(\F_q) \leq (q+1)^2 - \max(q-\delta_X + 1, 1)  \max(q-\delta_T + 1, 1).\]
\end{itemize}
Moreover each upper bound is reached by Proposition \ref{atteint}.
\end{corollary}

These upper bounds cannot be compared to Hasse-Weil. Indeed the curves that reach these bounds can be highly reducible and singular, as displayed in Proposition \ref{atteint}. Such a phenomenon has already been observed on general toric surfaces by I. Soprunov and J. Soprunova \cite{SS}.

\section{Punctured codes}\label{punc}

J. P. Hansen \cite{HanToric} and B. L. De La Rosa Navarro and M. Lahyane \cite{ DLR} studied codes, only on $\F_q$ with $q$ big enough so that the evaluation map is injective, that turn to be punctured codes of our evaluation code $C_\eta(\delta_T,\delta_X)$.

\medskip

In \cite{DLR}, the authors in fact considered a punctured code of $C_\eta(\delta_T,0)$ at $q+1$ coordinates as the evaluation code of polynomials of bidegree $(\delta_T,0)$ outside the fiber $\mathcal{F}$ for $q \geq \delta_T$. They obtained a quite bad puncturing since here the code $C_\eta(\delta_T,0)$ has parameters $[N,k,d]=[(q+1)^2,\delta_T+1,(q+1)(q-\delta_T+1)]$ whereas theirs has parameters $[N-(q+1),k, d-(q+1)]$.

\medskip

Among other toric surfaces, J.P. Hansen \cite{HanToric} and, more recently, J. Little and H. Schenck \cite{Little}, studied toric codes on Hirzeburch surfaces that evaluate polynomials of $R(\delta_T,\delta_X)$ for $\delta_T$ and $\delta_X$ both positive and only on $\F_q$-rational points of the torus $\Gb_m^2$. They also assumed that $\delta=\delta_T+\eta \delta_X < q-1$, which ensures the evaluation map to be injective. They proved the minimum distance to be equal to $(q-1)^2 - \delta(q-1)$.

Such a code is obtained from puncturing $C_\eta(\delta_T,\delta_X)$ at the $4q$ rational points of $\mathcal{Z}(T_1T_2X_1X_2)=D_1+D_2+E_1+E_2$.
They obtained a quite good puncturing since here the code $C_\eta(\delta_T,\delta_X)$ has parameters 
\[[N,k',d']=\left[(q+1)^2,(\delta_X+1)\left(\delta_T+\eta\frac{\delta_X}{2}+1\right),q(q-\delta_T+1)\right],\]
whereas theirs has parameters $[N-4q,k',d'-(3q-\delta-1)]$. Note that, as stated in the introduction, the difference between minimum distances is at least $2q$, the half of the difference between the lengths. This feature was already observed by G. Lachaud when extending Reed-Muller codes to projective Reed-Muller codes \cite{Lachaud}.

\medskip

We highlight here an interesting puncturing of codes $C_\eta(\delta_T,\delta_X)$ when $\delta_T$ is negative, in the sense that all common zero coordinates of codewords and only them are punctured. Let us define the linear code $C^\star_\eta(\delta_T,\delta_X)$ over $\F_q$ obtained by punctuation of the code $C_\eta(\delta_T,\delta_X)$ at the points of $\mathcal{Z}(X_1)=D_1$.

\begin{theorem}
Let $\eta \geq 1$, $\delta_T < 0$ and $\delta_X > 0$. The code $C^\star_\eta(\delta_T,\delta_X)$ has length $q(q+1)$ and has the same dimension and minimum distance as $C_\eta(\delta_T,\delta_X)$.
\end{theorem}

\begin{proof}
Every monomial $M=T_1^{c_1}T_2^{c_2}X_1^{d_1}X_2^{d_2} \in R(\delta_T,\delta_X)$ on $\mathcal{H}_\eta$ satisfies
\[0 \leq c_1+c_2=\delta_T + \eta d_1 < \eta d_1.\]
Hence $d_1 >0$ and $M$ is zero on $X_1=0$.
\end{proof}

\begin{remark}
The previous theorem is true even if $\eta =1$.
\end{remark}

\begin{example}
Here are some examples of punctured code $\F_3$, of length $12$ that reach the bounds given by code.tables \cite{codetables}.

\begin{center}
\begin{tabular}{|c|c|c|c|}
$\eta$ & $\delta_T$ & $\delta_X$ & Parameters of $C^\star_\eta(\delta_T,\delta_X)$\\
\hline
2 & -1 & 1 & [12,2,9] \\
2 & -1 & 3 & [12,10,2]\\
2 & -2 & 2 & [12,4,6] \\
2 & -2 & 3 & [12,8,3] \\
\end{tabular}
\end{center}
\end{example}

\section*{Acknowledgement}
The author expresses his gratitude to the anonymous referee for his careful work and
his helpful suggestions, especially for the third section.


\begin{thebibliography}{10}
\bibliographystyle{alpha}
\bibitem{CN}
Cicero Carvalho and Victor G.~L. Neumann.
\newblock Projective {R}eed-{M}uller type codes on rational normal scrolls.
\newblock {\em Finite Fields Appl.}, 37:85--107, 2016.

\bibitem{CD}
Alain Couvreur and Iwan Duursma.
\newblock Evaluation codes from smooth quadric surfaces and twisted {S}egre
  varieties.
\newblock {\em Des. Codes Cryptogr.}, 66(1-3):291--303, 2013.

\bibitem{CLO}
David~A. Cox, John Little, and Donal O'Shea.
\newblock {\em Ideals, varieties, and algorithms}.
\newblock Undergraduate Texts in Mathematics. Springer, Cham, fourth edition,
  2015.
\newblock An introduction to computational algebraic geometry and commutative
  algebra.

\bibitem{CoxToric}
David~A. Cox, John~B. Little, and Henry~K. Schenck.
\newblock {\em Toric varieties}, volume 124 of {\em Graduate Studies in
  Mathematics}.
\newblock American Mathematical Society, Providence, RI, 2011.

\bibitem{DLR}
Brenda~Leticia De~La Rosa~Navarro and Mustapha Lahyane.
\newblock Algebraic-geometric codes from rational surfaces.
\newblock In {\em Algebra for secure and reliable communication modeling},
  volume 642 of {\em Contemp. Math.}, pages 173--180. Amer. Math. Soc.,
  Providence, RI, 2015.

\bibitem{codetables}
Markus Grassl.
\newblock {Bounds on the minimum distance of linear codes and quantum codes}.
\newblock Online available at \url{http://www.codetables.de}, 2007.
\newblock Accessed on 2018-01-01.

\bibitem{HanToric}
Johan~P. Hansen.
\newblock Toric varieties {H}irzebruch surfaces and error-correcting codes.
\newblock {\em Appl. Algebra Engrg. Comm. Comput.}, 13(4):289--300, 2002.

\bibitem{Han}
S\o ren~Have Hansen.
\newblock Error-correcting codes from higher-dimensional varieties.
\newblock {\em Finite Fields Appl.}, 7(4):531--552, 2001.

\bibitem{Joyner}
David Joyner.
\newblock Toric codes over finite fields.
\newblock {\em Appl. Algebra Engrg. Comm. Comput.}, 15(1):63--79, 2004.

\bibitem{Lachaud}
Gilles Lachaud.
\newblock The parameters of projective {R}eed-{M}uller codes.
\newblock {\em Discrete Math.}, 81(2):217--221, 1990.

\bibitem{LP}
Christian Liedtke and Stavros~Argyrios Papadakis.
\newblock Birational modifications of surfaces via unprojections.
\newblock {\em J. Algebra}, 323(9):2510--2519, 2010.

\bibitem{Little}
John Little and Hal Schenck.
\newblock Toric surface codes and {M}inkowski sums.
\newblock {\em SIAM J. Discrete Math.}, 20(4):999--1014, 2006.

\bibitem{Reid}
Miles Reid.
\newblock Chapters on algebraic surfaces.
\newblock In {\em Complex algebraic geometry ({P}ark {C}ity, {UT}, 1993)},
  volume~3 of {\em IAS/Park City Math. Ser.}, pages 3--159. Amer. Math. Soc.,
  Providence, RI, 1997.

\bibitem{Ruano}
Diego Ruano.
\newblock On the parameters of {$r$}-dimensional toric codes.
\newblock {\em Finite Fields Appl.}, 13(4):962--976, 2007.

\bibitem{SS}
Ivan Soprunov and Jenya Soprunova.
\newblock Toric surface codes and {M}inkowski length of polygons.
\newblock {\em SIAM J. Discrete Math.}, 23(1):384--400, 2008/09.

\bibitem{Sturm}
Bernd Sturmfels.
\newblock {\em Gr\"obner bases and convex polytopes}, volume~8 of {\em
  University Lecture Series}.
\newblock American Mathematical Society, Providence, RI, 1996.

\end{thebibliography}
\end{document}